\documentclass[a4paper,11pt]{article}
\pdfoutput=1

\usepackage{amsmath, amssymb, mathrsfs, amsthm, bm, lipsum, cite,array}
\usepackage{mathtools}
\usepackage{appendix}
\usepackage{graphicx}
\usepackage{subfig}
\usepackage{longtable}
\usepackage{multirow}
\usepackage{pgf}
\usepackage{tikz} 
\usetikzlibrary{arrows,positioning,calc,decorations.pathmorphing}
\usepackage[colorlinks=true]{hyperref}
\hypersetup{linkcolor=black,citecolor=red}

\newtheorem{theorem}{Theorem}[section]
\newtheorem{lemma}[theorem]{Lemma}
\newtheorem{proposition}[theorem]{Proposition}

\theoremstyle{definition}

\theoremstyle{remark}

\newcommand{\drawsquare}[2]{\hbox{%
\rule{#2pt}{#1pt}\hskip-#2pt%  left vertical
\rule{#1pt}{#2pt}\hskip-#1pt%  lower horizontal
\rule[#1pt]{#1pt}{#2pt}}\rule[#1pt]{#2pt}{#2pt}\hskip-#2pt%  upper horizontal
\rule{#2pt}{#1pt}}% right vertical
\newcommand{\fund}{\raisebox{-.5pt}{\drawsquare{6.5}{0.4}}}
\newcommand{\antifund}{\overline{\fund}}

\newcommand\blfootnote[1]{%
  \begingroup
  \renewcommand\thefootnote{}\footnote{#1}%
  \addtocounter{footnote}{-1}%
  \endgroup
}
\newcolumntype{C}[1]{>{\centering\arraybackslash}m{#1}}
\newcommand{\comment}[1]{}

\begin{document}

\begin{titlepage}

\begin{flushright}
  USTC-ICTS-15-03\\
\end{flushright}

\begin{center}

\vspace{10mm}

{\LARGE \textbf{Seiberg Duality, Quiver Gauge Theories, \\
and Ihara's Zeta Function}}

\vspace{4mm}

{\large Da Zhou}$^{1,2}$,
{\large Yan Xiao}$^1$, and
{\large Yang-Hui He}$^{1,3,4}$

\blfootnote{
da.z.zhou@gmail.com \ , \quad
yan.xiao@city.ac.uk \ , \quad
hey@maths.ox.ac.uk
}

\vspace{1mm}

{\small
{\it
\begin{tabular}{rl}
${}^{1}$ &
Department of Mathematics, City University, London, EC1V 0HB, UK
\\
${}^{2}$ &
The Interdisciplinary Center for Theoretical Study,\\
&University of Science and Technology of China, Hefei, Anhui, 230026, P.R.~China
\\
${}^{3}$ &
School of Physics, NanKai University, Tianjin, 300071, P.R.~China
\\
${}^{4}$ &
Merton College, University of Oxford, OX14JD, UK
\end{tabular}
}}

\end{center}

\vspace{10mm}

\begin{abstract}
We study Ihara's zeta function for graphs in the context of quivers arising from gauge theories, especially under Seiberg duality transformations.
The distribution of poles is studied as we proceed along the duality tree, in light of the weak and strong graph versions of the Riemann Hypothesis.
As a by-product, we find a refined version of Ihara's zeta function to be the generating function for the generic superpotential of the gauge theory.
\end{abstract}

\end{titlepage}

\tableofcontents
%\listoffigures

\vspace{20mm}

%%%%%%%%%%%%%%%%%%%%%=========================================
\section{Introduction}
Quiver gauge theories have over the last two decades become fruitful in the cross-fertilization between physics and mathematics.
As quantum field theories, especially those with supersymmetry, they are archetypal of those arising from string theory and phenomenology; as finite graphs, they crystallize the underlying geometry and algebra, especially those with Calabi-Yau properties. More recently, the dialogue has extended to number theory, in particular to algebraic numbers and dessins d'enfants as well as to finite fields.

In \cite{He:2011ge}, computations of the zeta-function for finite graphs \cite{ihara} was initiated for quiver gauge theories bearing in mind the ultimate hope of finding the relation between properties such as whether the graph satisfies the analogue of the Riemann Hypothesis and the algebraic geometry of the Calabi-Yau moduli space of vacua. Perhaps the most remarkable action on quivers is mutation, where in the physics literature this is a guise of Seiberg duality \cite{2} and in the mathematics literature this is realized as cluster transformation \cite{FZ}. 
It is a beautiful fact that these were discovered independently, one as a duality between quantum field theories and another as an isomorphism of quiver representations (and Calabi-Yau geometry).
Thus it is natural to investigate the properties of the zeta function under such a key transformation.

There is therefore a trio of conversations: the physics of the quiver gauge theory, the geometry of the moduli space of representation (equivalently the vacuum moduli space of the field theory), as well as the number theory of the zero/pole structure of the graph zeta-function of Ihara.
It is this trio that we wish to analyse.

We begin, in \S\ref{sec:2}, with a recapitulation of the graph zeta function, for directed and undirected cases -- though most of our quivers are directed -- emphasizing on concepts such as regularity and Ramanujan.
In parallel, in \S\ref{s:quiver}, we briefly summarize the rudiments of the quivers and associated Seiberg dualities in the gauge theory.
After examining in detail case studies of two of the most famous quiver gauge theories, namely those arising from the world volume physics of D-branes in the back-ground of affine Calabi-Yau threefolds as cones over $\mathbb{C}\mathbb{P}^2$ and $\mathbb{C}\mathbb{P}^1 \times \mathbb{C}\mathbb{P}^1$, in \S\ref{s:p2} and \S\ref{s:f0}, we delve into a wealth of examples in \S\ref{s:eg}, in order to see how Seiberg duality influences the Riemann and Ramanujan properties of the Ihara zeta function and conversely how the latter can be harnessed as a tool to explore the duality tree of field theories. 

In \S\ref{s:numerical_riemann}, we apply a specific graph Riemann Hypothesis for our directed quivers and summarize the results in Table \ref{table:RH_exp}.
In due course, we find an interesting fact of how the zeta-function serves as a generating function for candidate terms in the superpotential by ``refining'' it with a multi-variable version in \S\ref{Zeta_Graph_Properties}. Finally, we conclude with discussion and prospects in \S\ref{s:conc}.

%%%%%%%%%%%%%%%%%=============================================
%%%%%%%%%%%%%%%%%=============================================
\section{The Ihara Zeta Function}\label{sec:2}
First, let us recall the definitions of Ihara's graph zeta function, first defined by Ihara \cite{ihara}, studied extensively in \cite{Horton,M-S,Terras,Tarfuleaa,Sunada:LF,Sunada:FG}, and introduced to the study of quiver gauge theories in \cite{He:2011ge}.
A few graph-theoretic concepts \cite{Terras}, most of which are self-explanatory, are needed before we give the full definition for the Ihara zeta function, which we include here for completeness and to set our nomenclature.

\begin{itemize}
	\item An {\bf undirected} graph $G = (V,E)$ is a finite non-empty set of \emph{vertices} $V$ and a finite multiset $E$ of undirected \emph{edges} or unordered pairs of vertices. 
A graph is  \emph{simple} if there is no loops, \emph{i.e.}, no edges of the form $(u,u) \; \forall u \in V$ and there is only a single edge between any two vertices.

	\item A {\bf directed} graph $G = (V,E)$ is a finite non-empty set of \emph{vertices} $V$ and a finite multiset $E$ of \emph{arrows} or ordered pairs of vertices. For an arrow $e = (u,v)$, $u$ is defined to be its \emph{origin} $o(e)$ and $v$ to be its \emph{terminus} $t(r)$ (sometimes these are also called head and tail, respectively. However, we will reserve the word tail for a usage to be introduce shortly). Its {\bf inverse} arrow $\bar{e}$ is defined to be $(v,u)$ by reversing its origin and terminus. In general, directed graphs do not need to have both arrows and its inverse arrows both present in $E$.
We make this distinction between undirected and directed graphs because the form of the zeta-function, as we shall see, is sensitively dependent thereupon.
In the context of the gauge theories which we will soon study, we will primarily, for the sake of chirality, study directed graphs.

	\item A {\bf cycle} $c$ of length $n$ in graph $G$ is a sequence $c = (e_1, e_2, \ldots ,e_n)$ of $n$ arrows in $G$ such that $t(e_i) = o(e_{i+1})$ for $1 \leq i \leq n-1$ and $t(e_n) = t(e_1)$. Cycle $c$ is said to have \emph{backtrack} if $\bar{e}_{i+1} = e_i$ for $1 \leq i \leq n-1$. In addition, $c$ is said to have \emph{tail} if $\bar{e}_n = e_1$.

    \item The \emph{in-degree} (respectively \emph{out-degree}) of any vertex of a graph is simply the number of in-coming (respectively out-going) arrows, on the other hand, the \emph{undirected degree} of a vertex is simply the number of length-1 walks starting from this vertex using undirected edges only.

	\item The \emph{$n$-multiple} of a cycle $c$ is the cycle formed by traversing $c$ $n$ times. A cycle is called \emph{primitive} if it is not some $n$-multiple of some other cycle for $n \geq 2$. For a cycle $c = (e_1, e_2, \ldots , e_n)$, the equivalence class $[c]$ is defined to be the cyclic permutations
	\[
	[c] = \{(e_1, e_2, \ldots , e_n), (e_2, e_3, \ldots , e_n, e_1), \ldots , (e_n, e_1, \ldots , e_{n-1})\}
	,\]
	which simply means cycles are equivalent up to choices of initial and terminal vertices. Therefore a {\bf prime} in a graph is primitive cycle that is non-backtracking, tailless and not a $n$-multiple cycle.

        \item We are primarily concerned with graphs which are \emph{finite}, \emph{i.e.}, finite number of nodes and edges) and \emph{connected}, \emph{i.e.} every node can be reached by traversing along some combinations of paths.

        \item The {\bf adjacency matrix} for a graph with $n$ nodes is an $n \times n$ matrix and has its $(i,j)$-th entry specifying the number of undirected edges from node $i$ to $j$ with $i$-th diagonal entry being twice the number of self-adjoining loops on $i$-th node. We emphasize that the adjacency matrix we use here has this convention for the diagonal; this has deep implications in tune with Cartan matrices, Frobenius eigenvalues and ADE classifications of Dynkin diagrams \cite{mckay,He:1999xj}.
Moreover, any undirected edges in a graph $G$ can be replaced with a bi-directional arrow, thereby assigning standard directions onto $G$. 
\end{itemize}
	
Now, we are readily for the central definition of our paper.
The {\bf Ihara zeta function}, also called {\bf graph zeta function}, was initially defined \cite{ihara} for graphs $G$ that are \emph{finite}, \emph{undirected}, \emph{connected} and \emph{tailless}. 
Nevertheless, the graph is allowed to have multiple edges between nodes as well as loops (length 1 cycle). 
With these conditions, the Ihara zeta function is defined as:
\begin{equation}
	\zeta_G(z)\coloneqq\prod_{[P]\in \mbox{Prime Cycles}} \left(1-z^{l([P])}\right)^{-1},
\end{equation}
where the infinite product is over all prime equivalence classes and $l([P])$ is the length of the prime cycle.
This definition is clearly motivated and parallels well that of the famous Riemann zeta-function, whose Euler-product is $\zeta(z) = \prod\limits_{p \in \mbox{{\tiny Primes}}}\left( 1 - p^{-z} \right)^{-1} = \sum\limits_{n \ge 1} n^{-z}$.

Indeed, as with all zeta-functions, the interplay between the expression as a product and as a sum is that between primality and integrality.
This is also the case with the graph zeta function.
It is shown in~\cite{Tarfuleaa} that the Ihara zeta function takes the following closed form for any partially directed graphs, in which both arrows and edges are allowed \cite{Horton,M-S}:
\begin{equation}
	\zeta_G(z)=\frac{(1-z^{2})^{-\rm{Tr}(Q-I)/2}}{{\rm Det}(I-Az+Qz^{2}+Pz^{3})},
	\label{eq:partial_zeta}
\end{equation}
where the following clarifications on notation are understood:
\begin{itemize} 
	\item This is the general equation for partially directed graphs, but we only have fully directed graphs for particular gauge theories.
	\item $A$ stands for the full adjacency matrix with its entry $a_{ij}$ 	specifying length 1 walk from $i$-th node to $j$-th node using either an edge or an arrow.
	\item $P$ is the directed adjacency matrix whose entries are composed entirely of arrows.
	\item $Q$ is the matrix for undirected degrees. The $i$-th undirected degree is the number of length $1$ walks from $i$-th node to its neighbours using undirected edges \emph{only}. Specifically, it turns out that the exponent in the numerator stands for the number of vertices minus the number of edges, or simply $1-r$ with $r$ being the rank of fundamental group of the graph. For nodes with self-adjoining loops, the undirected degree is counted as 2.
\end{itemize}
Lastly, the relation between Euler-product and determinant expression of of Ihara zeta function is discussed in more detail in Appendix \ref{ap:euler_determinant}.

%%%%%%%%%%%%%%%%%%%%%%%%%%%%%%%%%%%%%
\subsection{Poles, Regularity and Ramanujan}
Some of the initial motivations for defining the zeta function of graphs are, of course, number theoretical, and in particular in issues such as finding analogues of the Riemann Hypothesis (cf.~\cite{He:2015jla} for a recent equivalent restatement of the Riemann Hypothesis in string theory).
Here, let us highlight some salient features of Ihara's zeta function.

Now, a graph is called a $(q+1)$-\emph{regular graph} if all of its nodes are connected to other $q+1$ nodes through length one walks. In addition, a directed regular graph must satisfy a stronger condition that all the nodes must have their in-degrees and out-degrees equal (see page 29 of \cite{chen:1997}).
A $(q+1)-$regular graph is called {\bf Ramanujan} if the maximum of the absolute value of eigenvalues of the adjacency matrix $A$, excluding $q+1$ itself, is bounded by $2\sqrt{q}$. In short, 
\begin{equation}
{\rm max} \{ |\lambda|: \lambda \in {\rm Spec}(A), |\lambda| \neq q+1 \} \leq 2\sqrt{q} \ .
\label{ramanujan}
\end{equation}

While the original motivation to study Ramanujan graphs was because they have maximal gaps in their spectrum, it turns out, as we now see, that they play a key role in a graph version of the Riemann Hypothesis.
The Ihara zeta function of a $(q+1)-$regular undirected graph is said to satisfy \emph{Riemann Hypothesis} if, in complete analogy with the number-theoretic case, for $0<{\rm Re}(s)<1$, the zeros of $\zeta_G(q^{-s})^{-1}$ lie on the ${\rm Re}(s)=\frac{1}{2}$ line. Here, as is customary with definitions of zeta-functions, we define the exponentiated variable
\begin{equation}
z := q^{-s} \ .
\end{equation}
Note, however, the parameter $q$ is natural for regular graphs whereas for irregular ones more work is needed to extract a similar quantity. In addition to the above definition for Riemann Hypothesis for undirected regular graphs, it also can be shown \cite{Terras} that \emph{a $(q+1)$-regular undirected graph satisfies Riemann Hypothesis if and only if it is Ramanujan}.

Apart from the aforementioned property, Ramanujan graphs have diverse connections with various fields, such as expander graphs in Communication Network Theory that revolves around extremal problems, Number Theory, Representation Theory and Algebraic Geometry. For more comprehensive surveys, we refer the readers to \cite{Terras,Murty,Winnie:NT, Lubotzky,Valette}. A more recent survey concentrating on connection between Graph Theory and automorphic representation can be found in \cite{Winnie:RD}.

In more generality, we also have a definition of the {\it Graph Theory Riemann Hypothesis} for irregular undirected graphs as follows:
\begin{enumerate}
	\item $\zeta_G(z)$ is pole free for
	\begin{equation}
	R_G < |z| < \sqrt{R_G},
	\label{def:strong_RH}
	\end{equation}
	where $R_G$ is the radius of convergence of $\zeta_G(z)$, \emph{i.e.}, the position of nearest pole to origin.

        \item There is a weaker version of the above, requiring $\zeta_G(z)$ to be {\it pole free} in
	\begin{equation}
	R_G < |z| < \frac{1}{\sqrt{q}},
	\label{def:weak_RH}
	\end{equation}
where $q$ is the largest degree of vertices.
\end{enumerate}

These two definitions are motivated by the fact that if $z$ is substituted by $R_G^s$, all poles of $\zeta_G(z)$ are then located within the ``critical strip'', $0 < \rm{Re}(s) < 1$, with poles at $s=0$ (i.e., $z=1$) and $s=1$ (i.e., $z=R_G$). 
Importantly, as is central to the study of any zeta function, there is an underlying {\bf functional equation}.
It can be shown \cite{Terras} that the functional equation for Ihara's zeta function for undirected graphs can take the form
\begin{equation}	
\Lambda_G(z) := (1-z^2)^{r-1+\frac{n}{2}}(1-q^2z^2)^{\frac{n}{2}}\zeta_G(z) = (-1)^n\Lambda_G(\frac{1}{qz}),
	\label{Functional_Eqn}
\end{equation}
where $r$ is the rank of fundamental group with $n$ being the number of vertices. With substitution $z=q^{-s}$, it is obvious that the above equation is symmetric with respect to $s=\frac{1}{2}$.

Since irregular graphs do not have a functional equation relating $f(s)$ and $f(1-s)$, it is natural to define a pole free region for $\frac{1}{2} < \rm{Re}(s) < 1$, therefore motivating~\eqref{def:strong_RH}. The weak version in this sense simply shrinks the size of the pole free region. 

It should be noted that the formulation of Riemann Hypothesis in terms of Ihara zeta function is based on the fact that the adjacency matrix of an undirected regular graph is symmetric and we can find a functional equation $f(s)$ which is symmetric with respect to $s=\frac{1}{2}$. However, for the case we will consider shortly, especially in the context of Seiberg duality in $\mathcal{N}=1$ supersymmetric gauge theories, we will be primarily concerned only with directed graphs. Since the adjacency matrix in these theories are rarely symmetric and the form of Ihara zeta function is very different from the undirected case, it is not possible to construct a functional equation similar to that of~\eqref{Functional_Eqn}, thus we can not construct a direct analogue of Riemann Hypothesis for directed graphs.

%%%%%%%%%%%%%%%%%%%%%%%%%%%%%%%%%%%%%%%%%%%%%%%%%%%%%%%%%%%%%%%%%
%%%%%%%%%%%%%%%%%%%%%%%%%%%%%%%%%%%%%%%%%%%%%%%%%%%%%%%%%%%%%%%%%
\section{Four-Dimensional Quiver Gauge Theories}\label{s:quiver}
Having recapitulated the relevant ingredients from the study of graphs and Ihara's zeta function, we now turn to our protagonist, the finite (generically directed) graphs which are realized as {\bf quivers} associated to four-dimensional, ${\mathcal N}=1$ supersymmetric gauge theories.
Briefly, quivers are finite, labelled, directed graphs. In addition to these properties, they are also allowed to have loops, bi-directional arrows and self-adjoining loops. On the vertices we have gauge groups and the corresponding gauge theory has product gauge group in the from of $\prod U(N_{i})$ with $i$-th vertex contributing a single gauge group factor. In the infra-red, the $U(1)$ subgroups of each of the $U(N_i)$ factor become frozen and we have an effective gauge group of products of special unitary groups. However, from the quiver representation point of view, it is important to retain the $U(N_i)$ as labels of the nodes.
The matter content of the theories are represented as bi-fundamentals between gauge groups, so each arrow from $i$-th vertex to $j$-th vertex transforms as $(N_i,\overline{N}_j)$ representation of the factor gauge group $SU(N_{i})\times SU(N_{j})$. Gauge theories with only undirected edges in their quivers are \emph{non-chiral} and $i$-th vertex should be associated to a vector space $\mathbb{C}^{N_i}$ with arrows being maps in Hom($\mathbb{C}^{N_i}$,$\mathbb{C}^{N_j}$). 

Furthermore, the matter content of the theories must be \emph{anomaly free}. This condition ensures that the corresponding quantum field theory is well-defined. When translated into quivers, this constraint has the following form:
	\begin{equation}\label{anom}
		(a_{ij}-a_{ji})n_{i}=0,
	\end{equation}
where $a_{ij}$ is the \emph{full} adjacency matrix of the quiver with $n_{i}$ being the vertex rank vector encoding the list of integers $N_i$. This condition will become important when we discuss the zeta function of our quivers, and certain Diophantine equations dictating their possible rank assignments.

%%%%%%%%%%%%%%%%%%%%%%%%%%%%%%%%%%%%%%%%%%%%
\subsection{Seiberg Duality and Cluster Mutations}
\label{sec:seiberg-dual}
Modern studies of $\mathcal{N}=1$ supersymmetric gauge theory lead to Seiberg's discovery of IR duality between two QCD-like theories~\cite{2}, where the duals have same $N_{f}$ fundamental chiral flavors of quarks, but different gauge groups $SU(N_{c})$ and $SU({N_f}-{N_{c}})$. The duality predicts that the two theories have the same supersymmetric moduli space and 't Hooft anomaly matching condition. In extending the duality to other $\mathcal{N}=1$ theories~\cite{3,4,5}, the most suggestive method is via the Hannany-Witten suspended brane construction~\cite{6}. In such constructions, we obtain a D-dimensional gauge theory by extending a set of Dirichlet branes in D dimensions, \emph{i.e.} the embedding in the transverse $10-D$ dimensions determines the spectrum and other properties of the theory. Therein, the D-branes are strings probing the transverse space, with one or both ends on the NS5-branes. If we take $N_{c}$ finite length strings between 5-branes, a $U(N_{c})$ pure world-volume gauge theory is produced. By introducing $N_{f}$ strings with semi-infinite length, $N_{f}$ quarks are produced. 

In this stringy context, Seiberg duality is obtained through moving one 5-brane to exchange its position with others, while reversing the orientation of the finite strings. Within the context of $N_{f} \ge N_{c}$, the movement and reconnection of strings produce a dual theory with $N_{f}$ semi-infinite and $N_{f}-N_{c}$ strings~\cite{7}. Moreover, since the inverse coupling is identified with the length of the string, the duality arises if we vary the gauge coupling through infinity via variation of a moduli field.
Since the embedding of gauge theories into string theory is not unique, the field-theoretic duality is interpreted as different brane configurations that give dual low-energy physics.
\comment{
we have a large class of $\mathcal{N}=1$ cases supported on the D-branes probing singularity of $\mathbb{C}^{3}/\Gamma$ orbifolds. In particular, it is interesting to look for instances when $\Gamma \subset SU(3)$ if we want to preserve $\mathcal{N}=1$ SUSY when orbifolding the parent $\mathcal{N}=4$ super-Yang-Mills theory.}

Before we discuss Seiberg duals for specific theories, it is worthwhile to look at the rules for Seiberg transformation on gauge theories. 
We summarize that the prototype \cite{2} is a pair, within the conformal window
$\frac32 N_c \le N_f \le 3N_c$,
(1) direct electric theory with $N_c$ colours with $N_f$ flavours and (2) 
dual magnetic theory with $N_f-N_c$ colours also with $N_f$ flavours:
\begin{equation}
\begin{array}{cc}
\begin{array}{c||c|cc}
  & SU(N_c)& SU(N_f)_L & SU(N_f)_R \\ \hline  
Q & \fund & \fund & 1 \\ 
Q'& \antifund & 1 & \antifund \\ 
\end{array}
&
\begin{array}{c||c|cc}
  & SU(N_f-N_c)& SU(N_f)_L & SU(N_f)_R \\ \hline
q & \fund & \antifund & 1 \\ 
q'& \antifund & 1 & \fund  \\ 
M & 1 & \fund & \antifund 
\end{array}
\\ 
W=0
&
W=Mqq'
\end{array}
\end{equation}
In the above, the quarks $Q$ and $Q'$ are transformed to the dual quarks $q$ and $q'$ and a Seiberg dual meson $M$ together with superpotential $Mqq'$ is generated.
It was recognized in \cite{Feng:2000mi,Feng:2001bn,Cachazo:2001gh,Beasley:2001zp} that, when applied to one node of any ${\mathcal N}=1$ quiver, Seiberg duality is the following graphical rule:
\[
\includegraphics[trim=0mm 0mm 0mm 0mm, clip, width=4in]{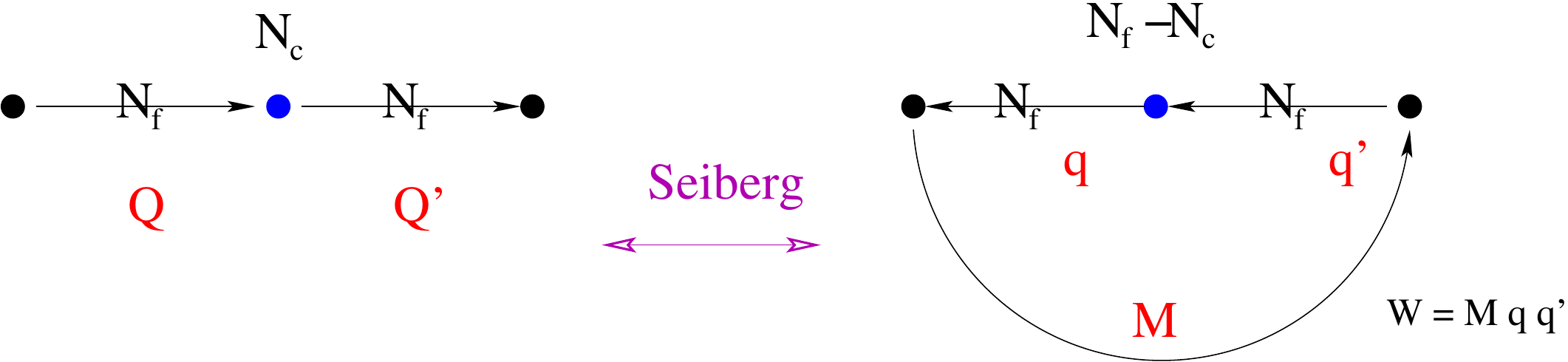}
\]
A truly remarkable fact is that around about the same time, mathematicians have independently noticed the importance of this graphical rule in the study of cluster variables \cite{FZ} and the above was called {\it cluster mutation} (though at the time, the role of the superpotential was not yet appreciated in the mathematics).
That Seiberg duality as an equivalence of quantum theory should be the same as cluster mutation generating an equivalence in the derived category of coherent sheafs on Calabi-Yau manifolds is a deep result in mathematical physics.

From an algorithmic point of view, one can think of the duality/mutation as the following transformation, for a quiver gauge theory with adjacency matrix $a_{ij}$ and $n$ vertices or gauge group factors, on $a_{ij}$ and rank vector $n_i$.
Indeed, there are $n$ choices to perform Seiberg transformation at any stage and the arrows and ranks will be changed while $n$ will always remain fixed.
Suppose, we dualized on node $i_0$, then the transformation is:
\begin{enumerate}
	\item Define $I_{in}$ to be the nodes having arrows coming into $i_0$, $I_{out}$ to be those having arrows coming from $i_0$ and $I_{no}$ to be those unconnected to $i_0$.
	\item Change the rank of node $i_0$ from $N_c$ to $N_f - N_c$, where $N_f = \sum_{i \in I_{in}}a_{i,i_0}N_i = \sum_{i \in I_{out}}a_{i_0,i}N_i$. This changes the rank of the gauge group of the dualized node.
	\item $a_{ij}^{dual} = a_{ji}$ if either $i,j=i_0$. In field theory context, this translates to the statement that the quarks of $SU(N_f-N_c)$ gauge group are in complex conjugate representation to the quarks of original $SU(N_c)$ group.
	\item $a_{\alpha \beta}^{dual} = a_{\alpha \beta} - a_{i_0 \alpha}a_{\beta i_0}$ for $\alpha \in I_{out}$ and $\beta \in I_{in}$. This is equivalent to adding Seiberg meson into the dual theory. Any bi-directional arrow corresponds to a quadratic mass term and can be integrated out.
	\label{rule}
\end{enumerate}

%%%%%%%%%%%%%%%%%%%%%%%%%%%%%%%%%%%%%%
%%%%%%%%%%%%%%%%%%%%%%%%%%%%%%%%%%%%%%
\section{Illustrative Example: Cone over \texorpdfstring{$\mathbb{P}^2$}{P2}}
\label{s:p2}
Having presented the necessary mathematics in terms of graph zeta functions and associated statements pertaining to their Riemann Hypotheses, as well as the physics in terms of four-dimensional $\mathcal{N}=1$ quiver gauge theories and their Seiberg duality, it is illustrative to start with an archetypical example to see the interplay of the two sides.

%%%%%%%%%%%%%%%%%%%%%%%%%%%%
\subsection{$dP_0$: Cone over \texorpdfstring{$\mathbb{P}^2$}{p2}}
Perhaps the most well-studied quiver is the one  presented in Figure~\ref{fig:p2}. If we write the 3 multi-arrows as $X_i, Y_i, Z_i$, it has an accompanying superpotential $W = \epsilon_{ijk}X^i Y^j Z^k$.  
This quiver gauge theory comes from taking the quotient of $\mathcal{N}=4$ super-Yang-Mills theory by a $\mathbb{Z}_3$-subgroup of $SU(3)$ to preserve $\mathcal{N}=1$ supersymmetry. Geometrically, since the moduli space of the parent $\mathcal{N}=4$ SYM is simply the affine Calabi-Yau 3-fold $\mathbb{C}^3$, the moduli space here is $\mathbb{C}^3 / \mathbb{Z}_3$ as a complex cone over $\mathbb{P}^2$ (and whose resolution can be seen as the total space of the anti-canonical bundle $\mathcal{O}_{\mathbb{P}^2}$ over $\mathbb{P}^2$).
Thus, historically, the theory is called the $\mathbb{P}^2$ or $dP_0$ theory.
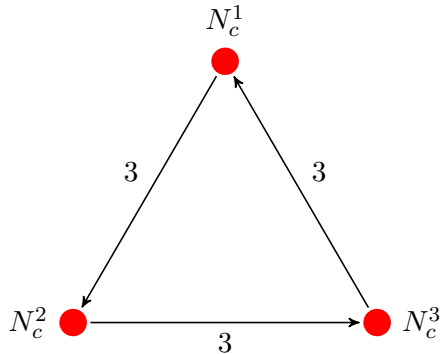
\begin{figure}[!ht]
\centering
	\begin{tikzpicture}[->,>=stealth',shorten >= 1pt,auto,semithick,
	main node/.style={circle,fill=red},
	inverse/.style={<-,shorten <= 1pt, auto, semithick}]
		\pgfmathsetmacro{\y}{4*sin(60)}
		\node[main node]		(1)		at (0,\y)		{};
		\node[black,above] 		at (1.north)		{$N_c^1$};
		\node[main node]		(2)  	at (-2,0)		{};
		\node[black,left] 		at (2.west)		{$N_c^2$};
		\node[main node]		(3)  	at (2,0)			{};
		\node[black,right] 		at (3.east)		{$N_c^3$};
			
		\path 	(2)		edge		[inverse]		node		{3}		(1)
				(3)		edge		[inverse]		node		{3}		(2)
				(1)		edge		[inverse]		node		{3}		(3);
			
	\end{tikzpicture}
	\caption{{\sf Quiver for $\mathbb{P}^{2}$, vertices have $SU(N_c)$ factor gauge groups and the number of arrows between each pair of vertices is 3.}}
	\label{fig:p2}
\end{figure}

%%%%%%%%%%%%%%%%%%%%%%%%%%%%%%%%%%%
\subsection{Adjacency Spectra and Poles}\label{sec:p2_pole}
Since this quiver is fully directed, its Ihara zeta function has a particular simple form~\cite{M-S,He:2011ge}:
\begin{equation}
	\zeta_G(z) = \frac{1}{{\rm Det}(I-Az)},
    \label{eq:fully-directed-zeta}
\end{equation}
where $A$ is the adjacency matrix for digraph $G$.
We see therefore that the reciprocal of zeta function for directed graphs is none other than the {\bf characteristic polynomial} of the adjacency matrix $A$,
\begin{eqnarray}
  \chi_A(\lambda) = {\rm Det}(\lambda I - A)
  \sim {\rm Det}(I - A \lambda^{-1})
  = \frac{1}{\zeta_G(\lambda^{-1})} \ .
  \label{eq:p2poleandspectra}
\end{eqnarray}
Therefore the behaviour of poles (those $z$ that make $\zeta_G(z)$ singular)
in Ihara zeta function is inversely proportional to the distribution of
eigenvalues (those $\lambda$ that make $\chi_A(\lambda)$ vanish)
of adjacency matrix for each Seiberg dual quiver.

Let us call the quiver in Figure \ref{fig:p2} $Q_0$ and hence the adjacency matrix and the Ihara zeta function are
\begin{equation} \label{eq:base-adjacency}
A_0 = \left( \begin{array}{ccc}
    0 & 3 & 0 \\ 0 & 0 & 3 \\ 3 & 0 & 0
  \end{array} \right) \ ,
\qquad
\zeta_{Q_0}(z) = \frac{1}{1 - 27z^3} \ .
\end{equation}
The above will be our ``\emph{basic case}'' from which we will perform repeated Seiberg duals on different nodes.
Using rules from Section~\ref{sec:seiberg-dual} to dualize on node $N_c^1$, we have following results as presented in Figure~\ref{fig:P2_dual}. 
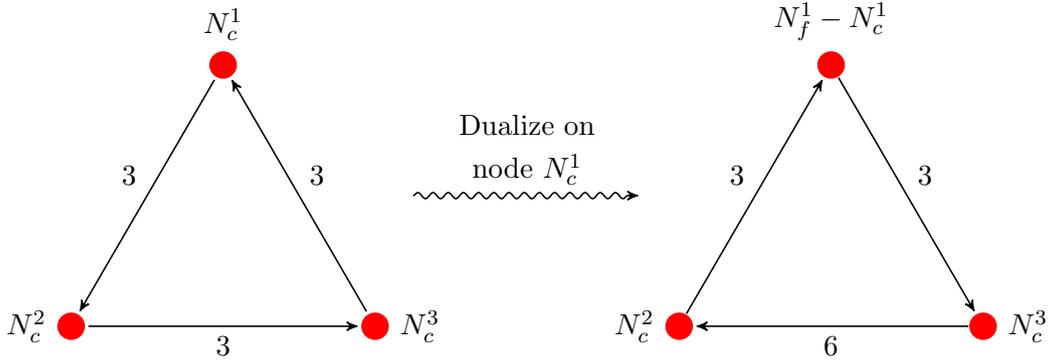
\begin{figure}[!h]
\centering
	\begin{tikzpicture}[->,>=stealth',shorten >= 1pt,auto,semithick,
	main node/.style={circle,fill=red},
	inverse/.style={<-,shorten <= 1pt, auto, semithick}]
		\pgfmathsetmacro{\y}{4*sin(60)}
		\pgfmathsetmacro{\yd}{2*sin(60)}
		\node[main node]		(1)		at (-4,\y)		{};
		\node[black,above] 		at (1.north)		{$N_c^1$};
		\node[main node]		(2)  	at (-6,0)		{};
		\node[black,left] 		at (2.west)		{$N_c^2$};
		\node[main node]		(3)  	at (-2,0)			{};
		\node[black,right] 		at (3.east)		{$N_c^3$};
			
		\path 	(2)		edge		[inverse]		node		{3}		(1)
				(3)		edge		[inverse]		node		{3}		(2)
				(1)		edge		[inverse]		node		{3}		(3);
		
		\draw [->,decorate,decoration={snake,amplitude=.4mm,segment length=2mm,post length=1mm}] (-1.5,\yd) -- (1.5,\yd)
		node [above,text width=3cm,align=center,midway]
		{Dualize on node $N_c^1$};
		
		\node[main node]		(4)		at (4,\y)		{};
		\node[black,above] 		at (4.north)		{$N_f^1 - N_c^1$};
		\node[main node]		(5)  	at (2,0)		{};
		\node[black,left] 		at (5.west)		{$N_c^2$};
		\node[main node]		(6)  	at (6,0)			{};
		\node[black,right] 		at (6.east)		{$N_c^3$};
			
		\path 	(5)		edge				node		{3}		(4)
				(6)		edge				node		{6}		(5)
				(4)		edge				node		{3}		(6);
	\end{tikzpicture}
	\caption{{\sf Quivers for $\mathbb{P}^{2}$ and its Seiberg dual on node $N_c^1$.}}
	\label{fig:P2_dual}
\end{figure}

The reciprocals of their zeta functions are readily seen to be $1-27z^3$
and $1-54z^3$ respectively. Hence the poles lie on the lines of third
roots of unity with $|z| = 1/\sqrt[3]{|c_3|}$ where $c_3$ is the coefficient
of $z^3$ term in the reciprocal. Let us perform duality to a few levels, picking any of the three nodes each time.
This has been well-known \cite{Franco:2003ja,Franco:2003ea,Hanany:2012mb} to exhibit a dendritic behaviour and can be drawn as a {\bf duality tree} where each node in the branching corresponds to a new quiver (i.e., this is a tree in the space of theories).
For reference, the tree for $dP_0$ is included in part (a) of Figure \ref{fig:P2-tree-poles}, where the same colours correspond to the same quiver, reflecting the 3-fold symmetric in the problem.

%%%%%%%%%%%%%%%%%%%%%%%%%%%%%%%%%%%%%%
\begin{figure}[!ht]
(a)
  \includegraphics[width=.42\textwidth]{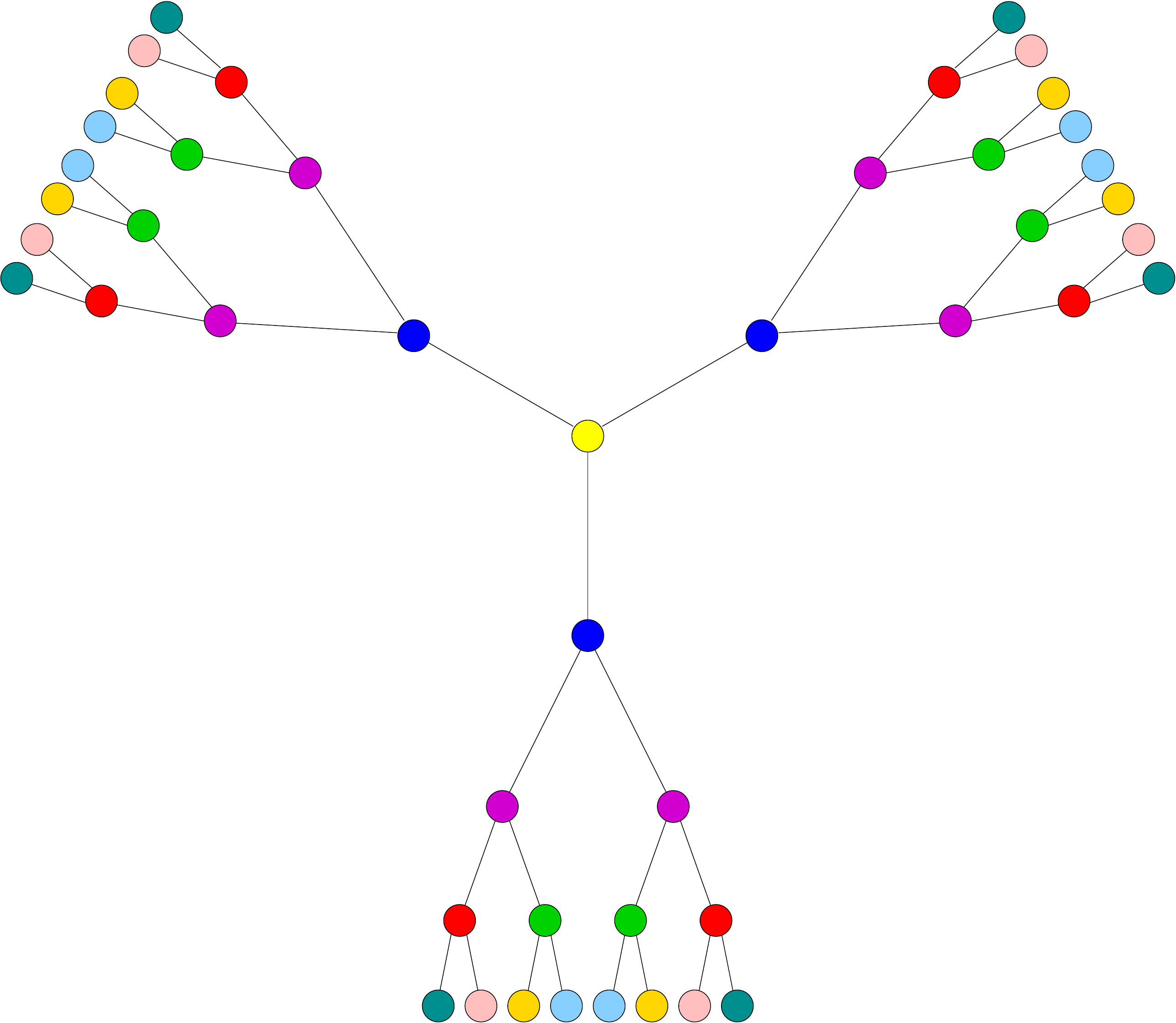}
(b)
  \includegraphics[width=.5\textwidth]{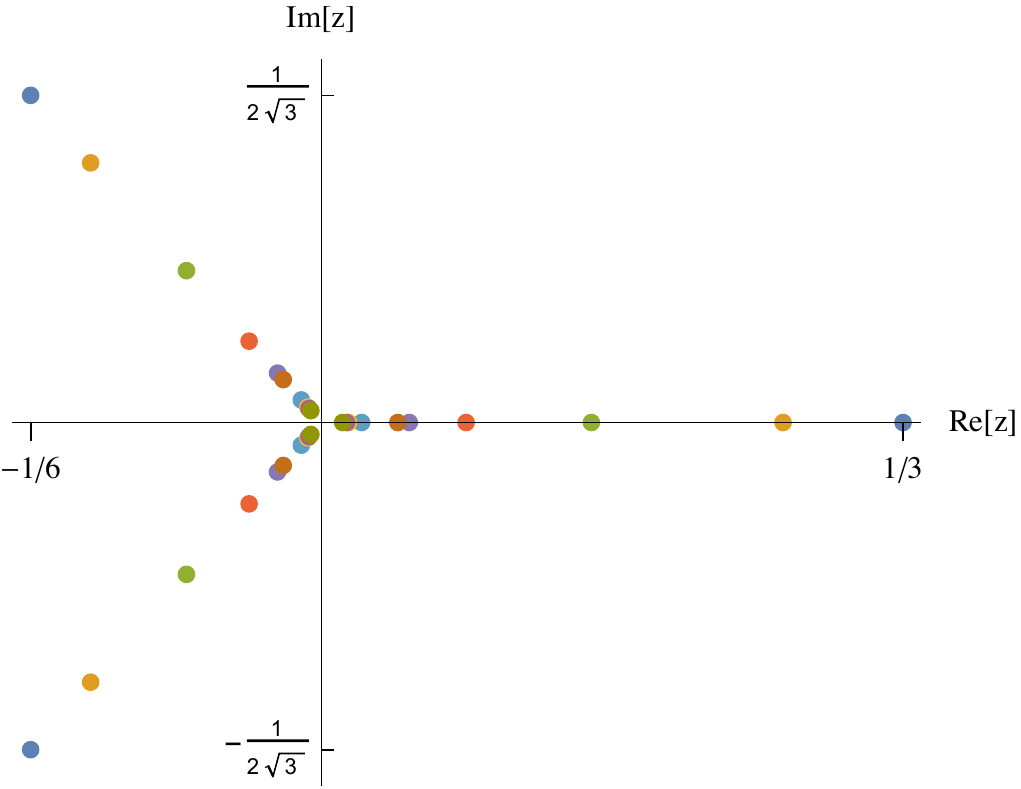}
  \caption{\sf $\mathbb{P}^2$ duality tree and the position of the poles on the complex plane for the Ihara zeta function for all the dual quivers, up to 5 levels.
NB: Colors of leaves in part (a) and the colors of dots in part (b) are not related. As we delve deeper into the duality tree, the lengths of poles, on a
  whole, get shorter, as the poles accumulate toward the origin.
This can be seen from Table \ref{tab:pole-length}.}
  \label{fig:P2-tree-poles}
\end{figure}

Now, let us compute the poles of the zeta function, for the various dual phases.
Numerically, we see that their length (i.e., distance to the origin on the complex plane) are as shown in Table \ref{tab:pole-length}.
\begin{table}[h!!!]
  \centering
  \begin{tabular}{|c|c|}
    \hline
    Level in duality tree & All possible pole lengths \\
    \hline
    1 & 1/3 \\
    \hline
    2 & 0.26 \\
    \hline
    3 & 0.15 \\
    \hline
    4 & 0.083, 0.050 \\
    \hline
    5 & 0.044, 0.016, 0.014, 0.0084 \\
    \hline
    6 & 0.023, 0.0048, 0.0040, 0.0024, 0.0014, 0.0012, 0.00080, 0.00043 \\
    \hline
  \end{tabular}
  \caption{{\sf Lengths of poles in first 6 levels. By length we simply mean the distance of the pole to the origin in the complex plane.}}
  \label{tab:pole-length}
\end{table}
In fact, all the Seiberg duals have their
zeta function poles sitting on the lines of cube roots of unity for 
$\mathbb{P}^2$
quiver (see Figure \ref{fig:P2-tree-poles}) and this can be derived
analytically as follows.

If we denote $(u,v,w)$ as the number of bi-fundamental arrows between the 3 nodes, a tree of branching integer triplets can be obtained by successive dualization. 
These triplets in turn describe all possible solutions of the $\mathbb{P}^2$ quiver under Seiberg duality. Furthermore, for all dual solutions to be anomaly free, the Diophantine equation~\cite{Cecotti:1992rm,Feng:2002kk,Franco:2003ja,Franco:2003ea,Benvenuti:2004dw,Hanany:2012mb}:
\begin{equation}\label{eq:dio1}
	u^2+v^2+w^2=uvw
\end{equation}
has to be satisfied. This is the famous \emph{Markov equation} and it is
well-known that all its solutions can be generated from the basic solution
$(3,3,3)$ by performing \emph{ad infinitum} the following three transformations
\cite{cassels:1965} (the proofs are gathered in Appendix~\ref{ap:markov} for the readers' reference),
\begin{eqnarray}
  (u,v,w) & \rightarrow & (u,v,uv-w) \ ,\nonumber\\
  (u,v,w) & \rightarrow & (u,uw-v,w) \ ,\nonumber\\
  (u,v,w) & \rightarrow & (vw-u,v,w) \ .
  \label{eq:transform-solution}
\end{eqnarray}

One can see that the above transformations are similar to Seiberg
transformations performed on quiver arrows (see rule 3 and 4 at the
end of Section~\ref{sec:seiberg-dual}). To be specific, the Seiberg
transformation rule on arrow numbers are,
\begin{eqnarray}
  (u,v,w) & \rightarrow & (-u,-v,w-uv) \ ,\nonumber\\
  (u,v,w) & \rightarrow & (-u,v-uw,-w) \ ,\nonumber\\
  (u,v,w) & \rightarrow & (u-vw,-v,-w) \ ,
  \label{eq:transform-arrow}
\end{eqnarray}
which differ from (\ref{eq:transform-solution}) only by a common factor $-1$.
This common factor can be interpreted as a reversal in the arrow directions.
So we can say, \emph{a Seiberg transformation generates a new quiver whose arrow
number combination is another solution to Markov equation} which means this new
quiver theory is automatically anomaly free. Since (\ref{eq:transform-solution})
generates all possible solutions, we can conclude that \emph{Seiberg transformation
generates all possible anomaly free dual quiver theories}. 

On the other hand,
since Seiberg transformation reverses arrows, if we choose the basic quiver
to be $Q_0$ whose cycles are all anti-clockwise as shown in Figure~\ref{fig:p2},
then all its descendants have their cycles being either clockwise or anti-clockwise.
Due to this fact all adjacency matrix $A$ will have one of the following two forms,
\begin{eqnarray}
\left( \begin{array}{ccc}
	0 & u & 0 \\ 0 & 0 & v \\ w & 0 & 0
	\end{array} \right)
\ , \qquad
\left( \begin{array}{ccc}
	0 & 0 & u \\ v & 0 & 0 \\ 0 & w & 0 .
	\end{array} \right)
\label{eq:adjacency}
\end{eqnarray}
In either case we have the same Ihara Zeta function $\zeta(z)$,
\begin{eqnarray}
	\zeta^{-1}(z) = {\rm det}(I-Az) = 1 - (uvw) z^3 .
\end{eqnarray}
Thus the poles of $\mathbb{P}^2$ Zeta function are
\begin{eqnarray}
  \frac{1}{\sqrt[3]{uvw}} ,\quad \frac{1}{\sqrt[3]{uvw}} \, e^{\pm i \frac{2\pi}{3}} .
  \label{eq:p2-poles}
\end{eqnarray}
This explicitly demonstrates that \emph{all poles after Seiberg transformation
lie on lines of third roots of unity} (see Figure~\ref{fig:P2-tree-poles}).
Moreover, since we have all the information on poles for all Seiberg transformed
quivers in this particular case, according to equation~\eqref{eq:p2poleandspectra},
we can see that all adjacency spectrum of each transformed quiver should also
lie on the same lines as the poles. 

It is interesting to remark that since
the original $\mathbb{P}^2$ quiver already has arrows between every pairs of
nodes, Seiberg transformation does not connect any previously unconnected
nodes to generate any new loops of length other than 3. As a consequence
of this, the reciprocal of Ihara Zeta functions of all duals only
have cubic terms in the complex variable. Indeed this property of Ihara
Zeta function is rather general and how it gathers information in terms of
simple cycles in a graph will be discussed in detail in Section~\ref{Zeta_Graph_Properties}.

%%%%%%%%%%%%%%%%%%%%
\subsection{Ramanujan Condition for $\mathbb{P}^2$ and its Seiberg Duals}

Due to the fact that Ihara zeta function only has the analogue of Riemann Hypothesis for undirected graphs, \emph{i.e.} for $(q+1)$-regular undirected graph $G$, $\zeta_G(z)$ satisfies the Riemann Hypothesis if and only if $G$ is Ramanujan, and quivers under consideration being fully directed graph, it is only possible to check if the graph is Ramanujan.
It should be noted that Ramanujan condition for directed graph has an extra constraint. In addition to the definition in section~\ref{sec:2}, the adjacency matrix of directed graphs should be diagonalizable by unitary matrices. See~\cite{12} for more details in the definition.

First of all, a quiver being Ramanujan requires this quiver to be a regular graph (that is, all vertices have the same in-degree and out-degree $d$). Since only the singular solutions to (\ref{eq:dio1}) can have same numbers and only the basic solution has its three numbers the same, obviously only the basic quiver is a regular graph. Thus none of the other quivers can be Ramanujan. We then check whether this basic quiver, denoted earlier as $Q_0$, is so.
We see that here all the eigenvalues $\lambda$ have $|\lambda|=3$, which means $Q_0$ cannot be Ramanujan and thus violates the graph Riemann hypothesis.

%%%%%%%%%%%%%%%%%%%%%%%%%%%%%%%%%%%%%%%%%%%%%%%%%
\section{Example: Cone over the Zeroth Hirzebruch Surface}\label{s:f0}
Another well-studied example in both the physics and mathematics literature (independently) is the quiver which corresponds to the cone $F_0$ over the zeroth Hirzebruch surface, $\mathbb{F}_0 := \mathbb{P}^1\times\mathbb{P}^1$, which is a toric variety.
The archtypal $F_0$ quiver (we shall call it the {\it basic} quiver) can be visualized in Figure \ref{fig:p1xp1}.
\begin{figure}[!ht]
\centering
	\begin{tikzpicture}[->,>=stealth',shorten >= 1pt,auto,semithick,
	main node/.style={circle,fill=green},
	inverse/.style={<-,shorten <= 1pt, auto, semithick}]
		\node[main node]		(1)		at (-2,4)		{1};
		\node[main node]		(2)  	at (-2,0)		{2};
		\node[main node]		(3)  	at (2,0)			{3};
		\node[main node]		(4)  	at (2,4)			{4};
			
		\path 	(2)		edge		[inverse]		node		{2}		(1)
				(3)		edge		[inverse]		node		{2}		(2)
				(4)		edge		[inverse]		node		{2}		(3)
				(1)		edge		[inverse]		node		{2}		(4);
			
	\end{tikzpicture}
    \caption{{\sf Basic quiver for $\mathbb{P}^1\times\mathbb{P}^1$, vertices are labeled with $1,2,3,4$.}}
	\label{fig:p1xp1}
\end{figure}
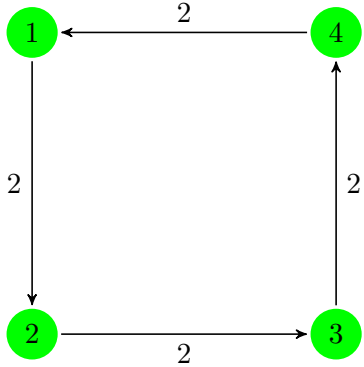

A series of quivers  -- we shall call them {\it general} $F_0$ quivers -- are generated by Seiberg dualizing on different vertices for a number of times and they are pictured in Figure \ref{fig:p1xp1-general}; indeed, this is the most general form of a 4-node quiver. We use the convention that a positive integer denotes a multiple of arrows while a negative means the same, but in the reverse direction. Perhaps the most famous dual quiver to the $(a,b,c,d,e,f) = (2,2,2,0,0,2)$ basic case is $(a,b,c,d,e,f) = (2,-2,-2,4,0,2)$ on dualizing on node 1, because both these quivers afford rank vectors which are all 1, i.e., they are both {\it toric}.

\begin{figure}[!ht]
\centering
	\begin{tikzpicture}[->,>=stealth',shorten >= 1pt,auto,semithick,
	main node/.style={circle,fill=green},
	inverse/.style={<-,shorten <= 1pt, auto, semithick}]
		\pgfmathsetmacro{\y}{4*sqrt(2)}
		\node[main node]    (1)    at (-2,4)    {1};
		\node[main node]    (2)    at (-2,0)    {2};
		\node[main node]    (3)    at (2,0)     {3};
		\node[main node]    (4)    at (2,4)     {4};
		\node[black,right]  (5)    at (-0.9,3)		{e};
		\node[black,right]  (6)    at (-0.9,1)		{d};
			
		\path 	(2)		edge		[inverse]		node		{c}		(1)
				(3)		edge		[inverse]		node		{a}		(2)
				(4)		edge		[inverse]		node		{f}		(3)
				(1)		edge		[inverse]		node		{b}		(4)
                (1)		edge		[inverse]		(3)
				(2)		edge		[inverse]		(4);
			
	\end{tikzpicture}
    \caption{{\sf General quiver emerging from dualizing the basic quiver.}}
	\label{fig:p1xp1-general}
\end{figure}
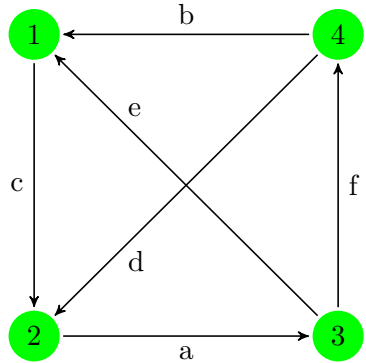

The antisymmetric adjacent matrix $q$ can be read off from this general quiver as
\begin{eqnarray}
  q = \left( \begin{array}{cccc}
    0 & c & -e & -b \\
    -c & 0 & a & -d \\
    e & -a & 0 & f \\
    b & d & -f & 0
  \end{array} \right).
  \label{eq:p1xp1-antisym-q}
\end{eqnarray}
On the one hand, to guarantee that the kernel of $q$ is not null so that one could have a sensible rank-vector according to anomaly cancellation in \eqref{anom}
we should impose the following condition (cf.~\cite{Hanany:2012mb})
\begin{eqnarray}
  \textrm{det}\,q = (ab+de-cf)^2 = 0 .
  \label{eq:no-null-kernel}
\end{eqnarray}
On the other hand, as a consequence of gauge anomaly cancellation,
a Diophantine equation is found \cite{Benvenuti:2004dw},
\begin{eqnarray}
  a^2 + b^2 + c^2 + d^2 + e^2 + f^2 + bcd + bef = ace + adf + abcf .
  \label{eq:p1xp1-dioph}
\end{eqnarray}
Physically one would expect every Seiberg dual theory of the basic quiver is an anomaly free theory,
i.e., each edge number combination in these dual quivers forms a solution to the
above two equations (\ref{eq:no-null-kernel}) and (\ref{eq:p1xp1-dioph}).

Yet, we seem to find that the converse proposition also holds. That is,
{\it all positive even solutions to these two equations can be generated by Seiberg duality starting from the basic quiver}.
By ``positive even solution'' we mean all variables in this solution
take positive and even integer values.
Indeed this a very complex set of multi-variable and non-linear Diophantine equations and we need to use (\ref{eq:no-null-kernel}) to get rid of one of the variables.
Firstly, it is easy to see that no solution can have $c=0$.
Suppose we have a solution with $c=0$, then $a$ or $b$ and $d$ or $e$ would have to be 0, the only possible non-trivial combination has to satisfy
\begin{equation}
b = c = e = 0 ,\quad a^2 + d^2 + f^2 = adf 
\ ,
\end{equation}
which reduces to the $\mathbb{P}^2$ quivers discussed above and is not the case we want to discuss here.
With this restriction, we have $f=(ab+de)/c$ from (\ref{eq:no-null-kernel}) and
equation (\ref{eq:p1xp1-dioph}) then reduces to a 5-variable Diophantine equation,
\begin{eqnarray}
  0 &=& a^2 b^2 + a^2 c^2 + b^2 c^2 + c^2 d^2 + c^2 e^2
  + d^2 e^2 + c^4 - a^2 b^2 c^2 + b c^3 d \nonumber\\
  && - a c^3 e + a b^2 c e - a^2 b c d + 2 abde
  - a c d^2 e + b c d e^2 - a b c^2 d e .
  \label{eq:p1xp1-reduced-dioph}
\end{eqnarray}

Instead of solving this equation analytically which is beyond our
reach we seek to find as many solutions as possible with computer power
and use these solutions to test our conjecture. The basic algorithm is
as follows.
\begin{itemize}
  \item[1] Generate all combinations of $(a,b,c,d,e)$ with any of the variables
    being a positive even integer within the range $[0, N]$;
  \item[2] Test each combination and select those that solve equation
    (\ref{eq:p1xp1-reduced-dioph}).
\end{itemize}
The result is, we have found all solutions with $N=800$
and all of them can be found in the theories of general quivers.
This is the analogue of the statement that Seiberg duality generates {\it all} solutions to the Markov equation and it would be nice to have an analytic proof of this fact.
For more detailed reference of our argument, we have put all results and the source code at
\href{https://github.com/dayzhou/Seiberg-Diophantine}{GitHub}\footnote{{\tt https://github.com/dayzhou/Seiberg-Diophantine}.}.

%%%%%%%%%%%%%%%%%%%%%%%%%%%%%%%%%%%%%%%%%%%%%%%%%%%%
\subsection{Ramanujan Condition for $F_0$ and its Seiberg Duals}
As with the $dP_0$ case, we can check explicitly if all the dual quivers generated from the basic $\mathbb{P}^1 \times \mathbb{P}^1$ quiver are Ramanujan.
First, let us check the basic case. The adjacency matrix there is
\begin{eqnarray}
Q_0 = \left( 
\begin{array}{cccc}
0 & 2 & 0 & 0 \\
0 & 0 & 2 & 0 \\
0 & 0 & 0 & 2 \\
2 & 0 & 0 & 0
\end{array}
\right) \ ,
\end{eqnarray}
and the poles are computed explicitly to be $\pm \frac{1}{2}$ and $\frac{1}{2} \, e^{\pm i \frac{\pi}{2}}$, and the modulus of all eigenvalues are thus seen to be $|\lambda| = 2$. Therefore the basic case is again not Ramanujan.

For the Seiberg duals generated from $\mathbb{P}^1 \times \mathbb{P}^1$, it is also checked in~\ref{eq:p1xp1-reduced-dioph} that the only regular graphs are those with the same Ihara zeta function, $(1-16z^4)^{-1}$. This is equivalent to say that the only regular graphs in $\mathbb{P}^1 \times \mathbb{P}^1$ and its duals are $\mathbb{P}^1 \times \mathbb{P}^1$ itself and the one that has all its arrow in opposite direction compared to $\mathbb{P}^1 \times \mathbb{P}^1$.

~\\

%%%%%%%%%%%%%%%%%%%%%%%%%%%%%%%%%%%%%%%%%%%%%
%%%%%%%%%%%%%%
%%%%%%%%%%%%%%%%%%%%%%%%%%%%%%%%%%%%%%%%%%%%%%
\section{Zeta Poles and Gauge Theories: A Plethora of Examples}\label{s:eg}
With our experience now in the $dP_0$ and $F_0$ examples, we can proceed to study the plethora of quiver gauge theories known to the literature.
We will present pole plots generated from Seiberg duals of $F_0$, $dP_1$, $dP_2$ and $dP_3$ quivers. Since an analytical result for poles of $dP_0$ and its duals was obtained in section~\ref{sec:p2_pole}, it is omitted here.
These Hirzebruch and del Pezzo theories \cite{Feng:2000mi} have since become canonical examples of studying (toric) quiver gauge theories and for the reader's convenience we tabulate the relevant (basic toric) quivers and their associated Ihara zeta functions here (cf.~\cite{He:2011ge}) in Table~\ref{table:toric_zeta}.

%%%%%%%%%%%%%%%%%%%%%%
\begin{center}
  \begin{longtable}{|C{2cm}|C{5cm}|C{5cm}|}
\hline

\multicolumn{1}{|c}{Gauge Theory} & \multicolumn{1}{c|}{Toric Phases}
& \multicolumn{1}{c|}{Ihara Zeta Function} \\
\endfirsthead

\hline
\multicolumn{1}{|c|}{Gauge Theory} & \multicolumn{1}{c|}{Toric Phases}
& \multicolumn{1}{c|}{Ihara Zeta Function} \\
\endhead

\multicolumn{1}{|c|}{} & \multicolumn{2}{r|}{\it\footnotesize{Continued on next page $\rightarrow$}} \\
\hline
\endfoot

\endlastfoot

\hline

$\mathbb{P}^2$ 						& \begin{tikzpicture}[->,>=stealth',shorten >= 1pt,auto,inner sep = 0,semithick,
	main node/.style={circle,fill=red},
	inverse/.style={<-,shorten <= 1pt, auto, semithick}]
		\pgfmathsetmacro{\y}{2*sin(60)}
		\pgfmathsetmacro{\yd}{sin(60)}
		\node[main node]		(1)		at (-2,\y)		{A};
		\node[main node]		(2)  	at (-3,0)		{B};
		\node[main node]		(3)  	at (-1,0)	 	{C};
			
		\path 	(2)		edge		[inverse]		node		{3}		(1)
				(3)		edge		[inverse]		node		{3}		(2)
				(1)		edge		[inverse]		node		{3}		(3);
				\end{tikzpicture} & $\zeta^{-1}(z) = 1-27z^3$ \\ \hline

\multirow{2}{*}{$\mathbb{P}^1 \times \mathbb{P}^1$} 	& 
\begin{tikzpicture}[->,>=stealth',shorten >= 1pt,auto,semithick,inner sep=0,
	main node/.style={circle,fill=red},
	inverse/.style={<-,shorten <= 1pt, auto, semithick}]
		\node[main node]		(1)		at (-1,2)		{A};
		\node[main node]		(2)  	at (1,2)		{B};
		\node[main node]		(3)  	at (1,0)		{C};
		\node[main node]		(4)  	at (-1,0)		{D};
			
		\path 	(1)		edge		node		{2}		(2)
				(2)		edge		node		{2}		(3)
				(3)		edge		node		{2}		(4)
				(4)		edge		node		{2}		(1);
			
	\end{tikzpicture} &  $\zeta^{-1}(z) = 1-16z^4$ \\ \cline{3-3}
& \begin{tikzpicture}[->,>=stealth',shorten >= 1pt,auto,semithick,inner sep=0,
	main node/.style={circle,fill=red},
	inverse/.style={<-,shorten <= 1pt, auto, semithick}]
		\node[main node]		(1)		at (-1,2)		{A};
		\node[main node]		(2)  	at (1,2)		{B};
		\node[main node]		(3)  	at (1,0)		{C};
		\node[main node]		(4)  	at (-1,0)		{D};
			
		\path 	(1)		edge		node		{2}		(2)
				(2)		edge		node		{2}		(3)
				(3)		edge		[inverse] node		{2}		(4)
				(4)		edge		[inverse] node		{2}		(1)
				(1)		edge 		[inverse] node 		{4}		(3);
			
	\end{tikzpicture}& $\zeta^{-1}(z) = 1-32z^3$\\ \hline
$dP_1$ 						& \begin{tikzpicture}[->,>=stealth',shorten >= 1pt,auto,semithick,inner sep=0,
	main node/.style={circle,fill=red},
	inverse/.style={<-,shorten <= 1pt, auto, semithick}]
		\node[main node]		(1)		at (-1,2)		{A};
		\node[main node]		(2)  	at (1,2)		{B};
		\node[main node]		(3)  	at (1,0)		{D};
		\node[main node]		(4)  	at (-1,0)		{C};
		\node[black,right]  	(5)    at (-0.5,1.6)		{1};
		\node[black,right]  	(6)    at (-0.5,0.4)		{1};
			
		\path 	(1)		edge		node		{1}		(2)
				(2)		edge		node		{2}		(3)
				(4)		edge		[inverse] node		{3}		(3)
				(1)		edge		[inverse] node		{2}		(4)
				(1)		edge 		node {} 	(3)
				(4)		edge 		node {} 	(2);
			
	\end{tikzpicture} &   $\zeta^{-1}(z) = 1-12z^3-12z^4$\\ \hline

\multirow{2}{*}{$dP_2$} 						& \begin{tikzpicture}[->,>=stealth',shorten >= 1pt,auto,inner sep = 0,semithick,
	main node/.style={circle,fill=red},
	inverse/.style={<-,shorten <= 1pt, auto, semithick}]
		\pgfmathsetmacro{\c}{2*cos(72)}
		\pgfmathsetmacro{\d}{2*cos(36)}
		\pgfmathsetmacro{\s}{2*sin(72)}
		\pgfmathsetmacro{\t}{2*sin(144)}
		\pgfmathsetmacro{\e}{1.2*cos(72)}
		\node[main node]		(1)		at (0,2)			{A};
		\node[main node]		(2)  	at (\s,\c)			{B};
		\node[main node]		(3)  	at (\t,-\d)	 		{C};
		\node[main node]		(4)  	at (-\t,-\d)	 	{D};
		\node[main node]		(5)  	at (-\s,\c)	 		{E};
		\node[black,right]  	(6)    at (0.5,\c)			{3};
		\node[black,right]  	(7)    at (-0.65,\c)		{1};
		\node[black,right]  	(8)    at (-1.3,\e)		{2};
		\node[black,right]  	(9)    at (0.3,-.3)		{2};

		\path 	(1)		edge		node		{1}		(2)
				(5)		edge		[inverse] node		{1}		(1)
				(1) 	edge 		node 		{}		(4)
				(2)		edge		node		{2}		(3)
				(3)		edge 		node 		{1} 	(4)
				(3) 	edge 		node 		{} 		(1)
				(4)		edge 		node 		{1} 	(5)
				(4) 	edge 		node 		{} 		(2)
				(5) 	edge 		node 		{} 		(3);
				\end{tikzpicture} & $\zeta^{-1}(z) = 1-16z^3-12z^4$\\ \cline{3-3}
& \begin{tikzpicture}[->,>=stealth',shorten >= 1pt,auto,inner sep = 0,semithick,
	main node/.style={circle,fill=red},
	inverse/.style={<-,shorten <= 1pt, auto, semithick}]
		\pgfmathsetmacro{\c}{2*cos(72)}
		\pgfmathsetmacro{\d}{2*cos(36)}
		\pgfmathsetmacro{\s}{2*sin(72)}
		\pgfmathsetmacro{\t}{2*sin(144)}
		\node[main node]		(1)		at (0,2)			{A};
		\node[main node]		(2)  	at (\s,\c)			{B};
		\node[main node]		(3)  	at (\t,-\d)	 		{C};
		\node[main node]		(4)  	at (-\t,-\d)	 	{D};
		\node[main node]		(5)  	at (-\s,\c)	 		{E};
		\node[black,right]		(6)		at(-.25,1.1)		{1};
		\node[black,right]		(7)		at(-1.1,.8)			{2};
		\node[black,right]		(8)		at(.9,-.6)			{1};
		\node[black,right]		(8)		at(0,-.65)			{1};

		\path 	(1)		edge		node		{1}		(2)
				(1)		edge		node		{}		(3)
				(2) 	edge 		node 		{}		(5)
				(2)		edge		[inverse] node		{1}		(3)
				(3)		edge 		[inverse] node 		{1} 	(4)
				(3)		edge 		node 		{} 		(5)
				(4) 	edge 		node 		{} 		(1)
				(4)		edge 		[inverse] node 		{2} 	(5)
				(5) 	edge 		node 		{1} 		(1);
				\end{tikzpicture}& $\zeta^{-1}(z) = 1 - 5z^3 - 12z^4 - 4z^5$\\ \hline

\multirow{4}{*}{$dP_3$}
& \begin{tikzpicture}[->,>=stealth',shorten >= 1pt,auto,inner sep = 0,semithick,
	main node/.style={circle,fill=red},
	inverse/.style={<-,shorten <= 1pt, auto, semithick}]
		\pgfmathsetmacro{\c}{2*sin(60)}
		\node[main node]		(1)		at (-2,0)			{F};
		\node[main node]		(2)  	at (-1,\c)			{A};
		\node[main node]		(3)  	at (1,\c)	 		{C};
		\node[main node]		(4)  	at (2,0)	 		{B};
		\node[main node]		(5)  	at (1,-\c)	 		{D};
		\node[main node]		(6)  	at (-1,-\c)	 		{E};
		\node[black,right]		(7)		at (1,.5)			{1};
		\node[black,right]		(8)		at (1,-.5)			{1};
		\node[black,right]		(9)		at (.3,-.5)			{1};
		\node[black,right]		(10)	at (-.6,-.5)		{1};
		\node[black,right]		(11)	at (-.9,0)			{2};
		\node[black,right]		(12)	at (-1.5,-.5)		{1};
		\node[black,right]		(13)	at (-.3,1.1)		{1};

		\path 	(1)		edge		[inverse] node		{1}		(2)
				(2)		edge		node		{1}		(3)
				(3) 	edge 		node 		{1}		(4)
				(4)		edge		node		{1}		(5)
				(5)		edge 		node 		{1} 	(6)
				(6)		edge 		[inverse] node 		{1} 	(1)
				(1)		edge 		node 		{}		(5)
				(2)		edge		node		{}		(4)
				(3)		edge		node		{}		(1)
				(4)		edge		node		{}		(6)
				(5)		edge		node		{}		(2)
				(6)		edge		node		{}		(2)
				(6)		edge		node		{}		(3);
				\end{tikzpicture} & $\zeta^{-1}(z) = 1 - 8z^3 - 12z^4-4z^5$\\ \cline{2-3}
&\begin{tikzpicture}[->,>=stealth',shorten >= 1pt,auto,inner sep = 0,semithick,
	main node/.style={circle,fill=red},
	inverse/.style={<-,shorten <= 1pt, auto, semithick}]
		\pgfmathsetmacro{\c}{2*sin(60)}
		\node[main node]		(1)		at (-2,0)			{F};
		\node[main node]		(2)  	at (-1,\c)			{A};
		\node[main node]		(3)  	at (1,\c)	 		{C};
		\node[main node]		(4)  	at (2,0)	 		{B};
		\node[main node]		(5)  	at (1,-\c)	 		{D};
		\node[main node]		(6)  	at (-1,-\c)	 		{E};
		\node[black,right]		(7)		at (1.05,0)			{1};
		\node[black,right]		(8)		at (.4,1)			{1};
		\node[black,right]		(9)		at (-.6,1)			{1};
		\node[black,right]		(10)	at (.4,-.8)			{1};
		\node[black,right]		(11)	at (-.6,-.8)		{1};
		\node[black,right]		(12)	at (-1.2,0)			{1};

		\path 	(1)		edge		[inverse] node		{1}		(2)
				(2)		edge		node		{1}		(3)
				(3) 	edge 		[inverse] node 		{1}		(4)
				(4)		edge		[inverse]node		{1}		(5)
				(5)		edge 		node 		{1} 	(6)
				(6)		edge 		[inverse] node 		{1} 	(1)
				(1)		edge 		node 		{}		(5)
				(3)		edge		node		{}		(1)
				(3)		edge		node		{}		(5)
				(4)		edge		node		{}		(2)
				(6)		edge		node		{}		(2)
				(6)		edge		node		{}		(4);
				\end{tikzpicture}& $\zeta^{-1}(z) = 1 - 2z^3 - 9z^4-6z^5$\\ \cline{2-3}
& \begin{tikzpicture}[->,>=stealth',shorten >= 1pt,auto,inner sep = 0,semithick,
	main node/.style={circle,fill=red},
	inverse/.style={<-,shorten <= 1pt, auto, semithick}]
		\pgfmathsetmacro{\c}{2*sin(60)}
		\node[main node]		(1)		at (-2,0)			{F};
		\node[main node]		(2)  	at (-1,\c)			{A};
		\node[main node]		(3)  	at (1,\c)	 		{C};
		\node[main node]		(4)  	at (2,0)	 		{B};
		\node[main node]		(5)  	at (1,-\c)	 		{D};
		\node[main node]		(6)  	at (-1,-\c)	 		{E};
		\node[black,right]		(7)		at (1,.8)			{2};
		\node[black,right]		(8)		at (-1.2,-.9)		{2};
		\node[black,right]		(9)		at (.6,-1)			{1};
		\node[black,right]		(10)	at (-.4,-.6)		{1};
		\node[black,right]		(11)	at (-1.4,.6)	{1};
		\node[black,right]		(12)	at (-1.6,-.4)		{1};

		\path 	(1)		edge		[inverse] node		{2}		(2)
				(2)		edge		[inverse] node		{1}		(3)
				(3) 	edge 		[inverse] node 		{1}		(4)
				(4)		edge 				  node		{1}		(5)
				(5)		edge 				  node 		{1} 	(6)
				(1)		edge 				  node 		{}		(3)
				(1)		edge				  node 		{}		(5)
				(2)		edge				  node 		{}		(4)
				(3)		edge				  node 		{}		(6)
				(5)		edge				  node 		{}		(2)
				(6)		edge				  node 		{}		(2);
              \end{tikzpicture} & $\zeta^{-1}(z) = 1 - 8z^3 - 16z^4$\\ \cline{2-3}
&\begin{tikzpicture}[->,>=stealth',shorten >= 1pt,auto,inner sep = 0,semithick,
	main node/.style={circle,fill=red},
	inverse/.style={<-,shorten <= 1pt, auto, semithick}]
		\pgfmathsetmacro{\c}{2*sin(60)}
		\node[main node]		(1)		at (-2,0)			{F};
		\node[main node]		(2)  	at (-1,\c)			{A};
		\node[main node]		(3)  	at (1,\c)	 		{C};
		\node[main node]		(4)  	at (2,0)	 		{B};
		\node[main node]		(5)  	at (1,-\c)	 		{D};
		\node[main node]		(6)  	at (-1,-\c)	 		{E};
		\node[black,right]		(7)		at (1,.8)			{2};
		\node[black,right]		(8)		at (-1.2,-.9)		{2};
		\node[black,right]		(9)		at (.6,-1)			{3};
		\node[black,right]		(10)	at (-.4,-.6)		{1};
		\node[black,right]		(11)	at (-1.4,.6)		{1};
		\node[black,right]		(12)	at (-1.6,-.4)		{1};

		\path 	(1)		edge		[inverse] node		{2}		(2)
				(2)		edge		[inverse] node		{3}		(3)
				(3) 	edge 		[inverse] node 		{1}		(4)
				(4)		edge 				  node		{1}		(5)
				(5)		edge 		[inverse] node 		{1} 	(6)
				(1)		edge 				  node 		{}		(3)
				(1)		edge				  node 		{}		(5)
				(2)		edge				  node 		{}		(4)
				(6)		edge				  node 		{}		(3)
				(5)		edge				  node 		{}		(2)
				(2)		edge				  node 		{}		(6);
				\end{tikzpicture} & $\zeta^{-1}(z) = 1 - 36z^3$ \\ \hline
\caption{{\sf Ihara zeta functions for various toric phases of del Pezzo and Hirzebruch quivers.}}
\label{table:toric_zeta}
\end{longtable}
\end{center}

While in the above table we have distinguished the toric phases because they are most popular, there is no reason to restrict to them here.
Indeed, we compute all the zeta functions along the duality tree for all the above geometries and study the pole structure thereof at once.
These are presented in Figure \ref{fig:pole_plots}.
%%%%%
\begin{figure}[!ht]
\begin{tabular}{cc}
\subfloat[Pole plot for $\mathbb{P}^1 \times \mathbb{P}^1$ quiver and its duals.]{\includegraphics[width=.5\textwidth,keepaspectratio=true]{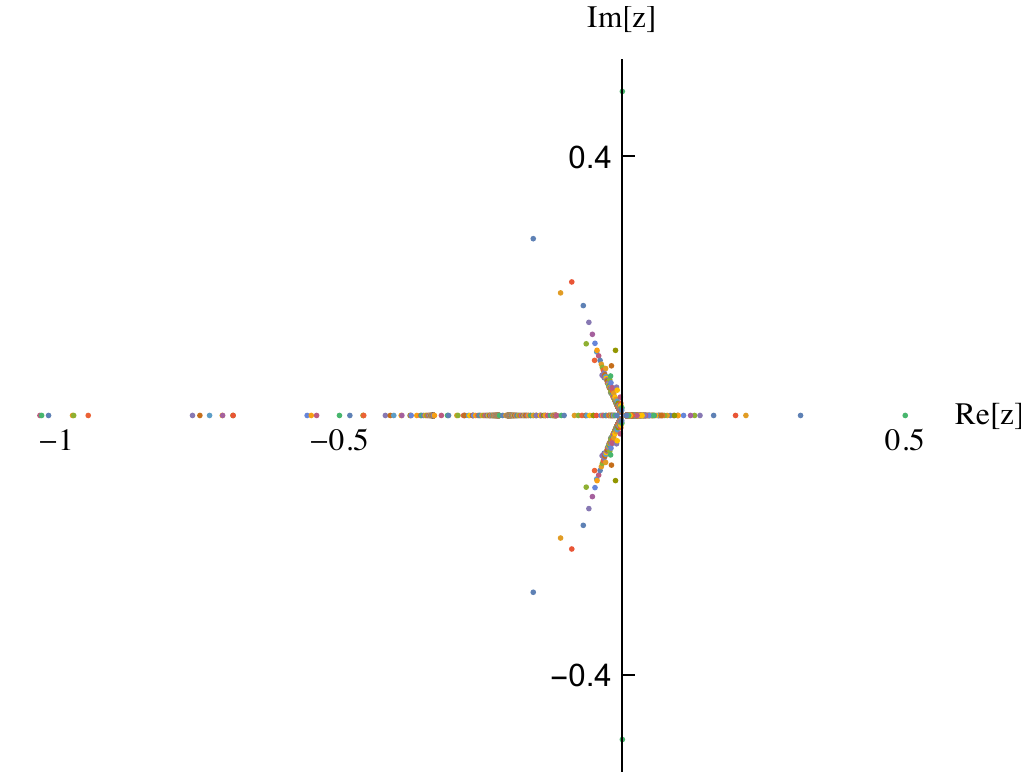}} &
\subfloat[Pole plot for $dP_1$ quiver and its duals.]{\includegraphics[width=.5\textwidth,keepaspectratio=true]{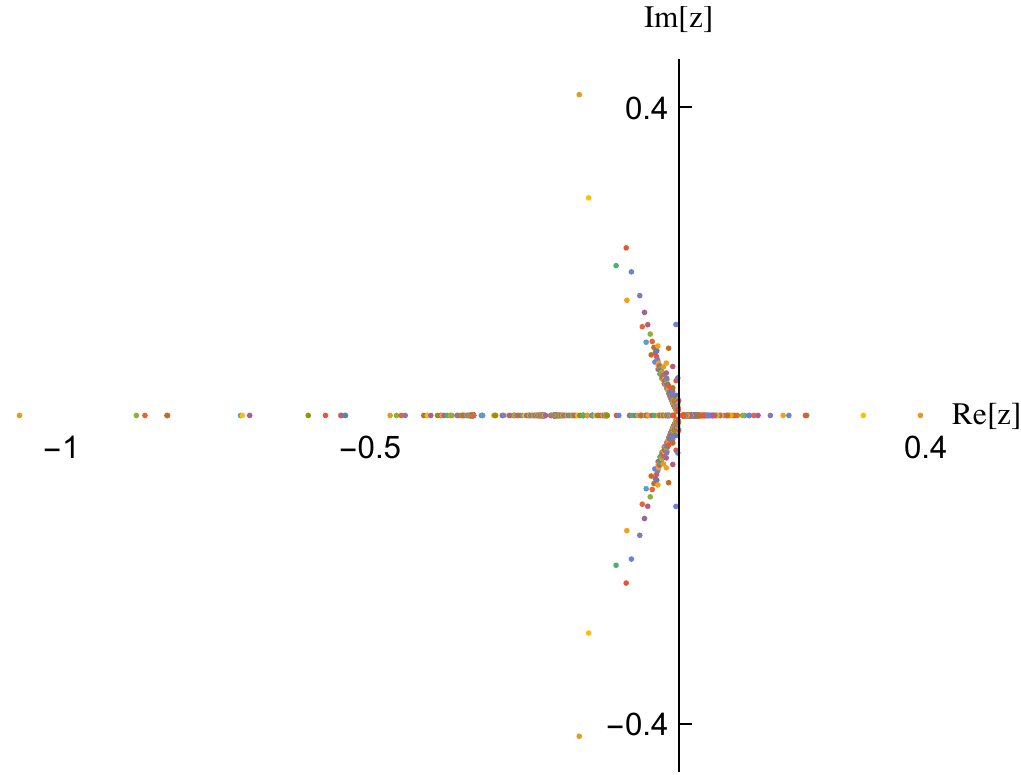}}\\
\subfloat[Pole plot for $dP_2$ quiver and its duals.]{\includegraphics[width=.5\textwidth,keepaspectratio=true]{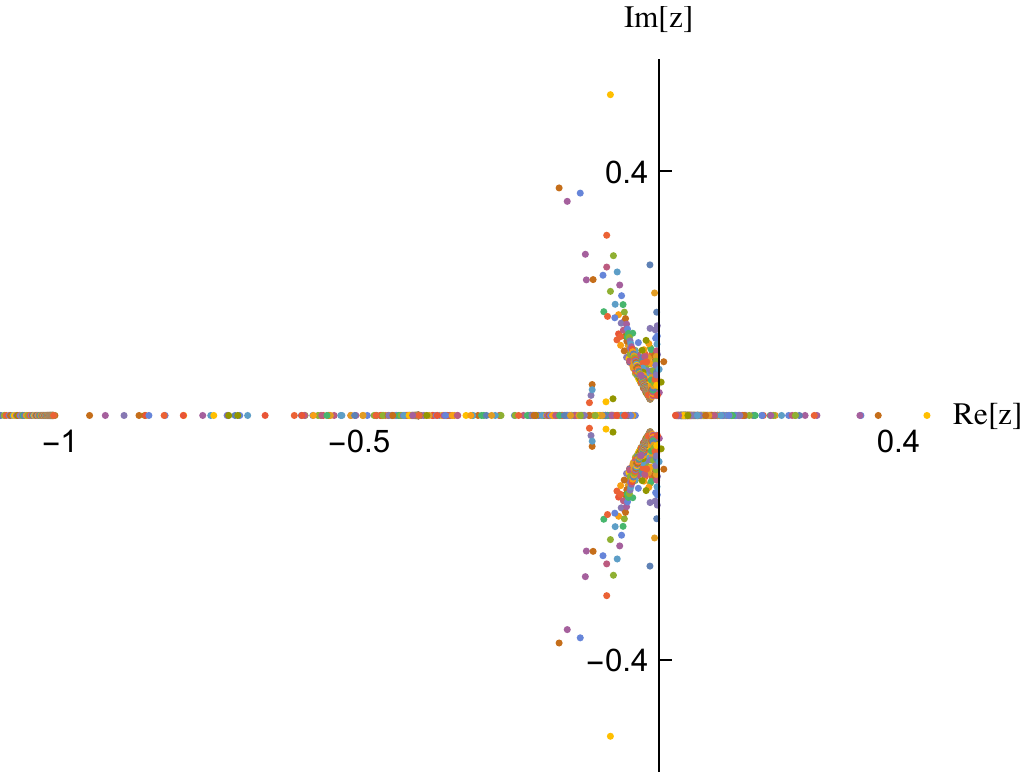}} &
\subfloat[Pole plot for $dP_3$ quiver and its duals.]{\includegraphics[width=.5\textwidth,keepaspectratio=true]{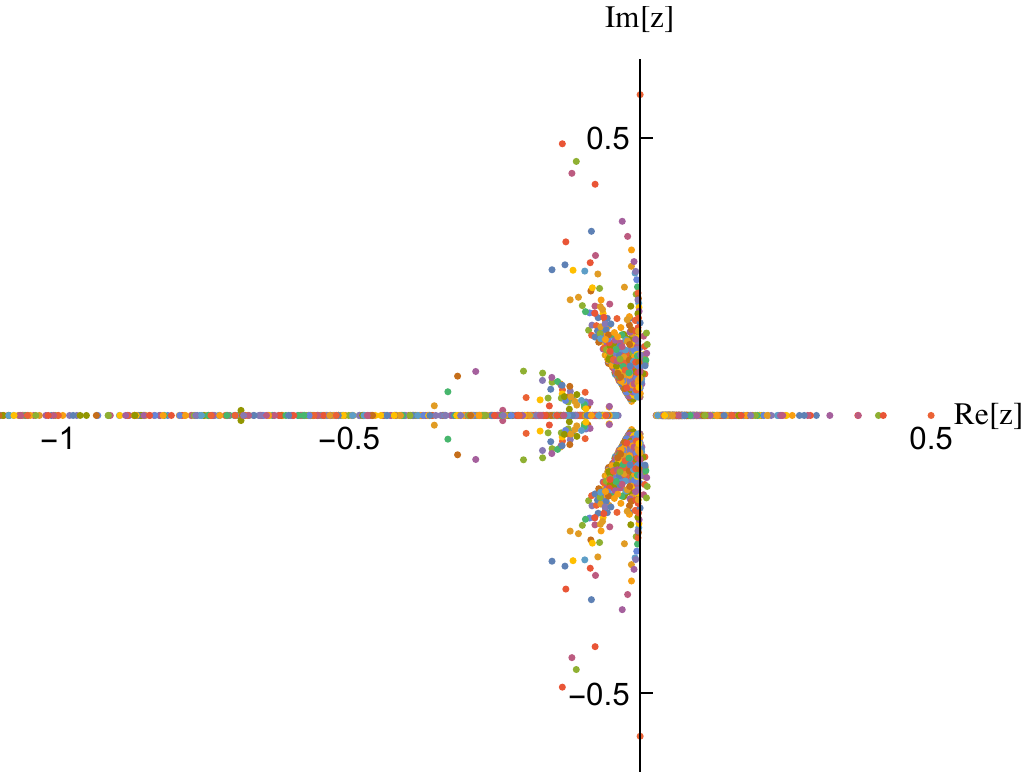}}
\end{tabular}
\caption{{\sf Ihara zeta function pole plots for $F_0$, $dP_1$, $dP_2$ and $dP_3$ quivers and their Seiberg duals.}}
\label{fig:pole_plots}
\end{figure}
%%%%%%%%%%%%
The figure is generated by dualizing on every vertex of the original quivers and then repeating the process for every dual quivers generated from the previous stage, following along the complicated {\sf duality trees} (See~\cite{Franco:2003ea}). Furthermore, due to the fact that the coefficients of the inverse Ihara zeta functions of dual quivers get very large very quickly (\emph{e.g.} for $\mathbb{P}^2$, Don Zagier conjectured that the $n$th Markov number is asymptotically given by $m_n = \frac{1}{3} e^{C\sqrt{n} + \mathcal{O}(1)}$ with $C = 2.3523\ldots$) we only look at dual quivers with arrow number smaller than 80.
 
The plots of $F_0$ and $dP_1$ are very similar and it can be explained due using results from section~\ref{Zeta_Graph_Properties}. Since the coefficients of the $z^n$ term in the inverse of Ihara zeta function expansion counts (we will return to discuss this point in detail in the following section) the number of length $n$ simple cycle if there does not exist smaller disjoint simple cycles with length $n_1$ and $n_2$ such that $n=n_1+n_2$. Hence in the $F_0$ and $dP_1$ dual quivers, the highest order term in the inverse of zeta function is $z^4$ as there are only 4 vertices in total and it not possible to form length-5 simple cycle or higher length simple cycle using disjoint lower length simple cycles. Therefore the inverse of their Ihara zeta functions should all have the form $1+az^3+bz^4$, where $a,b \in \mathbb{Z}$. The solutions to these equations therefore reflects the shape of pole plots of $F_0$ and $dP_1$. Similarly, $dP_2$ can not have term of order higher than 5 in their expansions and $dP_3$ can not have term of order higher than 6 as well.

Since all of these toric phases can be obtained from Seiberg transformation, their poles are also plotted in figure~\ref{fig:pole_plots}. It is also interesting to note that the second Phase of $dP_3$ is the only regular graph in all of Seiberg duals of $dP_3$. However, its adjacency matrix is not diagonalizable by a unitary matrix, and it is therefore not Ramanujan.

%%%%%%%%%%%%%%%%%%%%%%%%%%%%%%%%%%%%%%%%%%%%%%%%%%%%%%%
\subsection{A Numerical Experiment on Graph Riemann Hypothesis}\label{s:numerical_riemann}
As we saw in the introductory sections that the Ramanujan property is intimately linked to the graph Riemann Hypothesis.
In this subsection, we shall explore certain possibilities in applying both strong and weak Graph Riemann Hypothesis (RH) to our toric phases of del Pezzo quivers. Even though there is no immediate definition of RH pertaining to directed graphs due to our lack of an functional equation with reflection symmetry about $s =\frac{1}{2}$, it is still interesting to see how certain setups of RH for undirected graphs select special theories from Seiberg duals of del Pezzo quivers. In this perspective, we have the following setup:
\begin{enumerate}
\item The degree $q + 1$ is defined to be the maximum degree among all in-degrees and out-degrees of all vertices. This definition is in analogous spirit to that of Eq.~\eqref{def:weak_RH}, where the definition of weak RH is given for undirected irregular graph.
\item Dual quivers are selected if they satisfy either strong \emph{or} weak RH.
\item We then look at distinct prime factors for each term in $\zeta^{-1}$ for each dual quivers satisfying either two versions of RH. Especially if they form any interesting prime sequences.
\item Finally we look at how RH selects dual quivers in the duality tree. 
%The case of $\mathbb{P}^2$ is omitted here since all dual quivers trivially satisfy RH in this particular setup.
\end{enumerate}

First let us consider $\mathbb{P}^2$ case, since it was shown in section \ref{s:p2} that all of its Seiberg duals have poles lying on lines of cubic root of unity or having same modulus, we immediately see that both versions of RH are trivially satisfied. 
%%%%%%%%%%%%%
\begin{table}[htbp]
\centering
\begin{tabular}{|C{2.5cm}||C{1cm}|C{6.5cm}|C{4cm}|}
\hline
Base Quiver & Level &$\zeta^{-1}$ RH Duals& Prime Sequence \\
\hline
$\mathbb{P}^1 \times \mathbb{P}^1$ & 5 &\shortstack{$1-16 z^4,1-32 z^3,$\\$1-96 z^3,1-1056 z^3,$\\$1-480 z^3,1-14432 z^3,$\\$1-2720 z^3,1-28320 z^3,$\\$1-68640 z^3,1-200736 z^3$ }& $z^3$: 2,3,5,11,13,17,41,59\\
\hline
$dP1$ & 5 &None & N/A \\
\hline
\shortstack{$dP_2$ \\Toric Phase 1 }& 5 &\shortstack{$1-35 z^3-80 z^4-25 z^5,$\\$1-528 z^3-3082 z^4-2480 z^5,$\\$1-366 z^3-1594 z^4-700 z^5,$\\$1-984 z^3-8536 z^4-12040 z^5$} & \shortstack{$z^3$:  2,3,5,7,11,41,61\\$z^4$: 2,5,11,23,67,97,797\\$z^5$: 2,5,7,31,43}\\
\hline
\shortstack{$dP_2$ \\ Toric Phase 2 }& 5 &\shortstack{$1-35 z^3-80 z^4-25 z^5,$\\$1-366 z^3-1594 z^4-700 z^5,$\\$1-528 z^3-3082 z^4-2480 z^5,$\\$1-732 z^3-5328 z^4-5994 z^5$} & \shortstack{$z^3$:  2,3,5,7,11,61\\$z^4$: 2,3,5,23,37,67,797\\$z^5$: 2,3,5,7,31,37}\\
\hline
\shortstack{$dP_3$ \\Toric Phase 1 }& 4 &\shortstack{$1-36 z^3,1-25 z^3-61 z^4-30 z^5-4 z^6,$\\$1-8 z^3-64 z^4,1-72 z^3,$\\$1-180 z^3,1-12 z^3-92 z^4-40 z^5,$\\$1-173 z^3-523 z^4-96 z^5,1-792 z^3,$\\$1-205 z^3-806 z^4-531 z^5-90 z^6,$\\$1-77 z^3-169 z^4-90 z^5$} & \shortstack{$z^3$: 2,3,5,11,41,173\\$z^4$: 2,13,23,31,61,523\\$z^5$: 2,3,5,59\\$z^6$: 2,3,5}\\
\hline
\shortstack{$dP_3$\\ Toric Phase 2} & 5 &\shortstack{$1-36 z^3,1-25 z^3-61 z^4-30 z^5-4 z^6,$\\$1-8 z^3-64 z^4,1-72 z^3,$\\$1-180 z^3, 1-12 z^3-92 z^4-40 z^5,$\\$1-205 z^3-806 z^4-531 z^5-90 z^6,$\\$1-173 z^3-523 z^4-96 z^5,1-792 z^3$} & \shortstack{$z^3$: 2,3,5,11,143,173\\$z^4$: 2,13,23,31,61,523\\$z^5$: 2,3,5,59\\$z^6$: 2,3,5}\\
\hline
\shortstack{$dP_3$\\Toric Phase 3} & 4 & \shortstack{$1-36 z^3,1-180 z^3,$\\$1-25 z^3-61 z^4-30 z^5-4 z^6,$\\$1-8 z^3-64 z^4,1-72 z^3,$\\$1-320 z^3-1688 z^4-1612 z^5-280 z^6,$\\$1-12 z^3-92 z^4-40 z^5,$\\$1-173 z^3-523 z^4-96 z^5,$\\$1-2520 z^3,1-792 z^3$} & \shortstack{$z^3$: 2,3,5,7,11,173\\$z^4$: 2,23,61,211,523\\$z^5$: 2,3,5,13,31\\$z^6$: 2,5,7}\\
\hline
\shortstack{$dP_3$\\Toric Phase 4} & 3 & \shortstack{$1-30 z^3,1-186 z^3,$\\$1-9 z^3-18 z^4,1-6 z^3-12 z^4,$\\$1-51 z^3,1-4 z^3-14 z^4-6 z^5,$\\$1-6 z^3-6 z^4-12 z^5,$\\$1-2703 z^3,1-16 z^3-44 z^4-12 z^5,$\\$1-22 z^3,1-31 z^3-80 z^4-36 z^5,$\\$1-543 z^3,1-42 z^3-130 z^4-28 z^5,$\\$1-5 z^3-9 z^4-7 z^5-2 z^6$} & \shortstack{$z^3$: 2,3,5,7,11,\\17,31,53,181\\$z^4$: 2,3,5,7,11,13 \\$z^5$: 2,3,5,7\\$z^6$: 2}\\
\hline
\end{tabular}
\caption{{\sf Prime sequences in this table are produced for each term in $\zeta^{-1}$ expansion. Level specifies the level in duality tree up to which the dualization is performed.}}
\label{table:RH_exp}
\end{table}

We move on to the higher del Pezzo quivers and tabulate the relevant experimental results.
In Table \ref{table:RH_exp}, we see that toric phases have similar prime factor sequences for each term in $\zeta^{-1}$. This should be expected as toric phases can be reached by Seiberg transformation from each other. The difference in this case can be due to the fact that the dualization started from different quiver and if this process is continued long enough, sequences of different toric phases should converge to the same one. Specifically, the dual $\mathbb{P}^1 \times \mathbb{P}^1$ quivers that satisfy RH are those with $z^3$ terms only in its inverse Ihara zeta functions. This is also expected as those with only $z^3$ terms reduce to $\mathbb{P}^2$ case, which trivially satisfies both RHs.

It is more illustrative to look at how RH selects quivers in the $\mathbb{P}^1 \times \mathbb{P}^1$ duality tree, which we include from \cite{Franco:2003ea} for reader's convenience.

\begin{figure}[!ht]
\begin{tabular}{cc}
\subfloat[Duality tree for $\mathbb{P}^1 \times \mathbb{P}^1$ theory.]{\includegraphics[width=.5\textwidth,keepaspectratio=true]{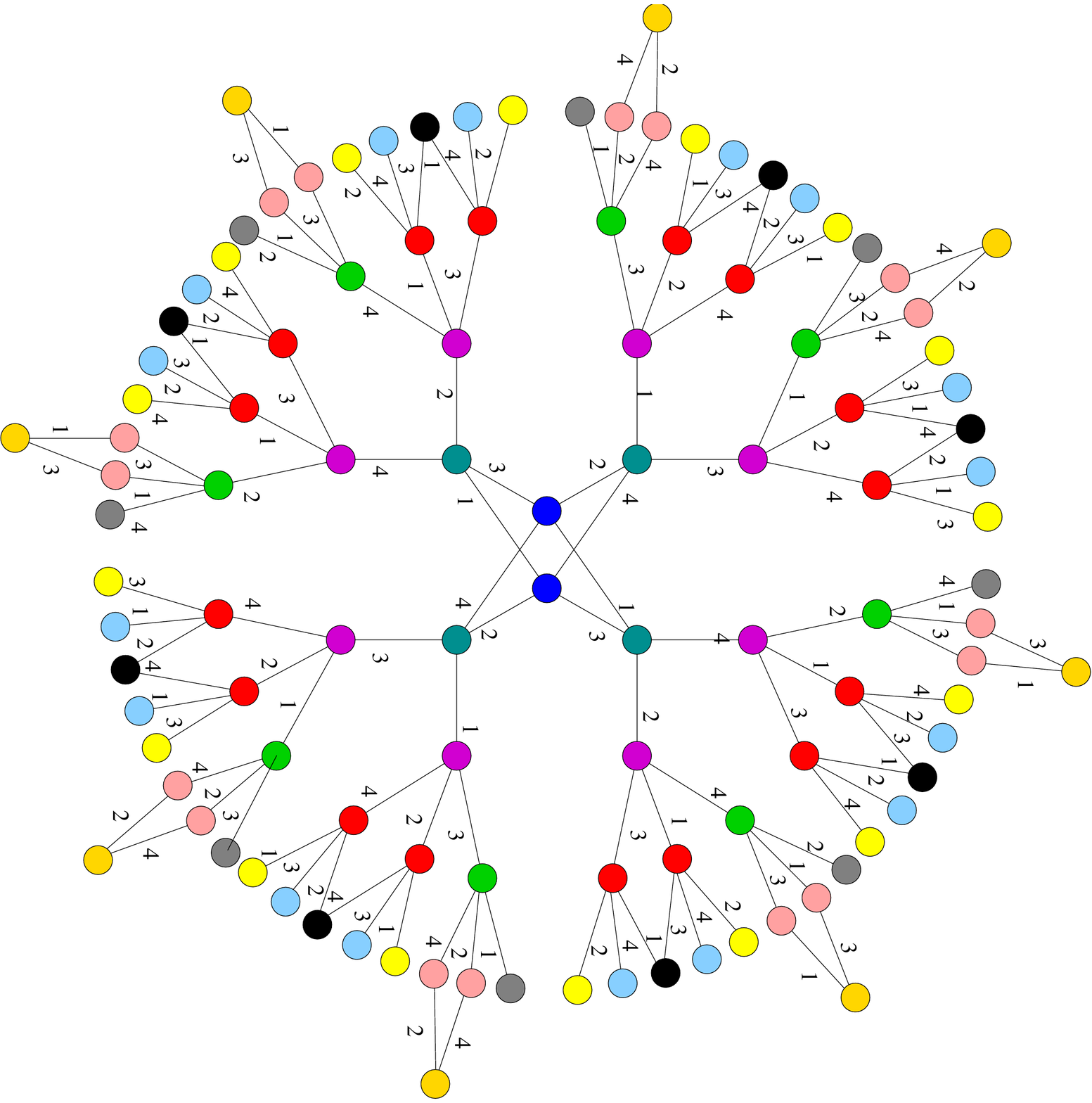}} &
\subfloat[Quiver Rules for $\mathbb{P}^1 \times \mathbb{P}^1$ duality tree.]{\includegraphics[width=.5\textwidth,keepaspectratio=true]{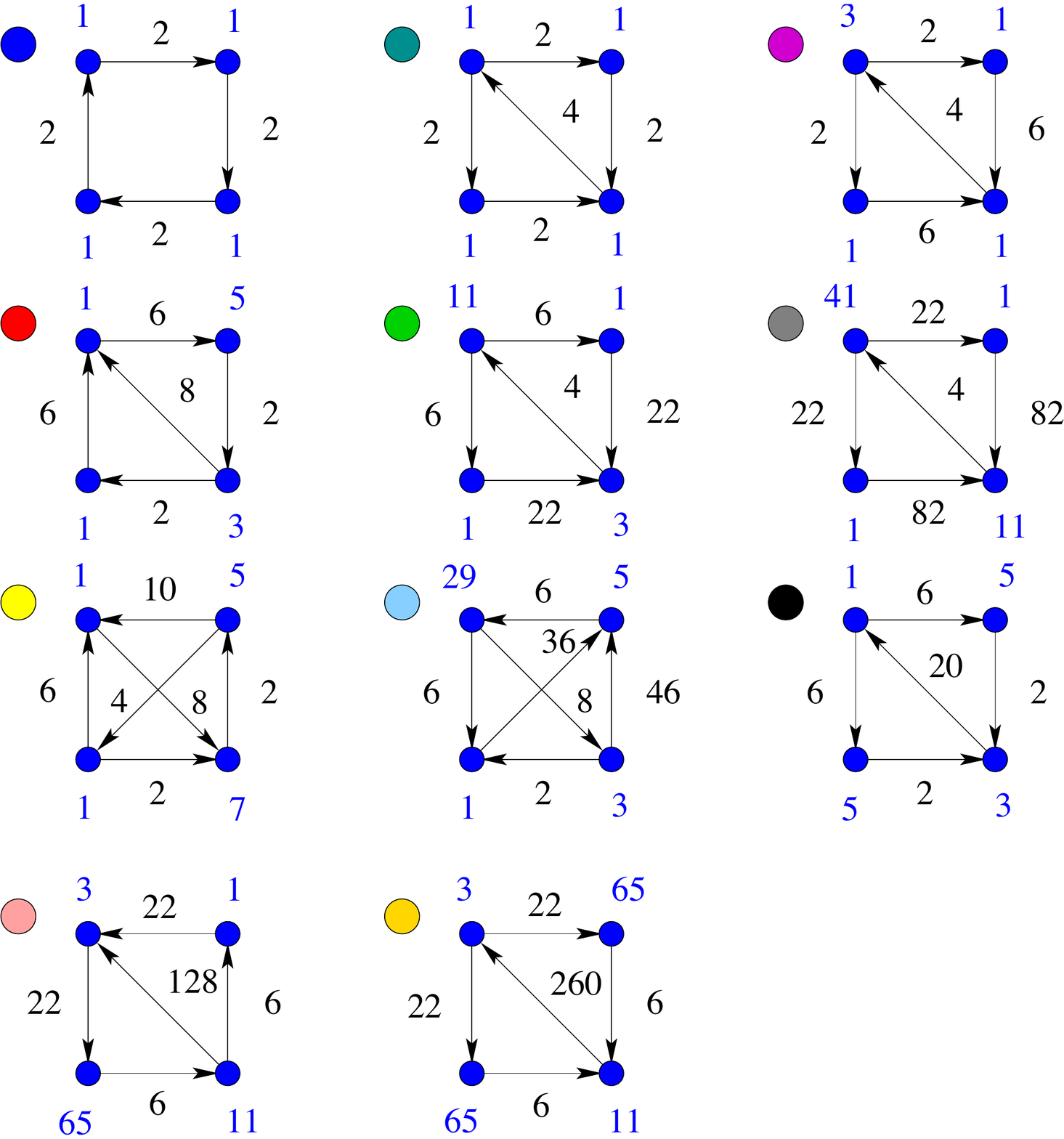}}
\end{tabular}
\caption{{\sf Duality tree and corresponding quiver rules for $F_0$.}}
\label{fig:P1xP1_duality_tree}
\end{figure}

Since we have all our $\mathbb{P}^1 \times \mathbb{P}^1$ dual quivers satisfying RH to contain terms in $z^3$ only in table \ref{table:RH_exp}. It can be seen from figure \ref{fig:P1xP1_duality_tree} that in our setup, RH only chooses the node without repetition of colors. In our plot, this corresponds to a path that takes the sequence of colors such as blue $\rightarrow$ cyan $\rightarrow$ purple $\rightarrow$ green $\rightarrow$ grey.
%%%%%
%%
%%%%%%
\section{Ihara Zeta Function, Quivers and Superpotentials}
\label{Zeta_Graph_Properties}

In this section we focus on a proof of a proposition presented below, which reinterprets coefficients of
the inverse of Ihara zeta function in terms of simple cycles. In terms of gauge theories, this translates to
generic superpotentials that can be generated from certain quivers.
\begin{proposition}
For a fully directed quiver $G$ with no self-adjoint loops, the reciprocal of its Ihara zeta function is the generator for simple loops in the sense that
\begin{equation}
 \zeta_G^{-1}(z) = \sum_{k=0}^n \left(
\sum_{\{j_1,\cdots,j_k\}} \sum_{\sum_i l_i=k} \prod_{a} (-[a])
\right) z^k \ ,
  \label{eq:final-ck}
\end{equation}
where $[a]$ is the number of ways to walk through a simple cycle in a particular
vertex sequence.
\label{thm:final-ck}
\end{proposition}

\begin{proof}
As we saw in the previous sections, 
the reciprocal of the Ihara zeta function for a fully directed graph $G$
with no self-adjoint loops is simply
\begin{eqnarray}
  \zeta_G^{-1}(z) = {\rm Det}(I_n-Az) ,
  \label{eq:reciprocal-zeta}
\end{eqnarray}
where $A$ is the adjacency matrix with entries denoted by $a_{ij}$.
Let us define $Z$ as the matrix in the above determinant expression,
\begin{eqnarray}
  Z := I_n - Az = \left( \begin{array}{cccc}
    1 & -a_{12}z & \cdots & -a_{1n}z \\
    -a_{21}z & 1 & \cdots & -a_{2n}z \\
    \vdots & \vdots & \ddots & \vdots \\
    -a_{n1}z & -a_{n2}z & \cdots & 1
  \end{array} \right) .
  \label{eq:matrix-Z}
\end{eqnarray}
Now, Eq.~(\ref{eq:reciprocal-zeta}) can be expanded as
\begin{eqnarray}
  \zeta^{-1} = {\rm Det} Z =
  \sum_{i_1\cdots i_n} \varepsilon_{i_1\cdots i_n} Z_{1i_1} \cdots Z_{ni_n} ,
  \label{eq:zeta-expansion}
\end{eqnarray}
where we have omitted the subscript $G$ and argument $z$ on the left for convenience.
Obviously the right hand side is a finite polynomial of $z$ in which no term
has its power higher than $n$. All terms have several 1's and several
$-a_{ij}z$ terms, therefore a general term with $(n-k)$ 1's would look like
\begin{eqnarray}
  \varepsilon_{i_1\cdots i_n} (-a_{j_1i_{j_1}}z) \cdots (-a_{j_ki_{j_k}}z)
  = (-1)^k \varepsilon_{i_1\cdots i_n} a_{j_1i_{j_1}} \cdots a_{j_ki_{j_k}} z^k .
  \label{eq:general-term}
\end{eqnarray}
Written in this way, we should note that we do not sum over all $i$-indices
any more. The reason is that once we have chosen some combination of
$\{j_1,\cdots,j_k\}$, say $\{1,\cdots,k\}$, then the 1's automatically take
other indices, i.e. $k+1$ to $n$. Since all 1's are coming from the diagonal
entries, which means $Z_{k+1,k+1}=\cdots=Z_{nn}=1$ or $i_{m}=m$ for $m>k$.
For this specific combination, this term simplifies to
\begin{eqnarray}
  (-1)^k \varepsilon_{i_1\cdots i_k(k+1)\cdots n} a_{1i_1} \cdots a_{ki_k} z^k
  = (-1)^k \varepsilon_{i_1\cdots i_k} a_{1i_1} \cdots a_{ki_k} z^k ,
  \label{eq:specific-term}
\end{eqnarray}
which indicates only $i_1$ to $i_k$ are summed indices. We can then rewrite
(\ref{eq:zeta-expansion}) as
\begin{eqnarray}
  \zeta^{-1} &=& \sum_{k=0}^n c_k z^k ,\nonumber\\
  c_k &=& (-1)^k \sum_{\{j_1,\cdots,j_k\}} \sum_{\{i_{j_1},\cdots,i_{j_k}\}}
  \varepsilon_{\cdots i_{j_1}\cdots i_{j_k}\cdots}\ a_{j_1i_{j_1}} \cdots a_{j_ki_{j_k}} ,\nonumber\\
  &=& (-1)^k \sum_{\{j_1,\cdots,j_k\}} \sum_{\{i_1,\cdots,i_k\}}
  \varepsilon_{i_1\cdots i_k} a_{j_1i_1} \cdots a_{j_ki_k} ,
  \label{eq:det-sum}
\end{eqnarray}
where the first summation in $c_k$ sums over all possible unordered combinations of
$\{j_1,\cdots,j_k\}$ and the second sums over all permutations of $\{j_1,\cdots,j_k\}$,
and it is understood that $\varepsilon_{i_1\cdots i_k}=1$ if $i_1<\cdots<i_k$. One can
easily see that $c_0=1$ from the above expression. Next we want to find out the
graphical meaning of other coefficients.

For clarity, we define some terms to make later analysis easier.
\begin{itemize}
  \item {\sf Simple Loop}: we call a monomial a simple loop if the indices of its
    components can be written in a cyclic pattern: $a_{ij}a_{jl}\cdots a_{pq}a_{qi}$
    and each index appears exactly twice, once as a row index and once as a column
    index. We say a monomial with its indices written in such cyclic form is in
    its {\it standard form}. We use square brackets to denote the standard form:
    \[ [a]_{ij\cdots li} := a_{ij} \cdots a_{li} , \]
    where the subscripts keep track of the vertex sequence in this loop.
  \item {\sf Equivalent Loops}: we say two loops are equivalent if their index sequences are different rotations of the same sequence given that both terms are written in standard forms. E.g., $a_{12}a_{23}a_{31}=a_{23}a_{31}a_{12}$. Obviously equivalent loops have equal values.
  \item {\sf Disjoint Loops}: we say two loops are disjoint if their corresponding
    monomials use different (disjoint) sets of indices, e.g., $a_{12}a_{23}a_{31}$
    and $a_{45}a_{54}$. In terms of graph theory, they do not share same vertices.
\end{itemize}
Now we claim that every term in the summation of (\ref{eq:det-sum}) can be split into
several (or 1) disjoint loops. This is due to the fact that for any term in the expansion
of a determinant each index should appear exactly twice. Then we can start from an
arbitrary $a_{ij}$ and select $a_{jk}$ to be the next entry. By repeating this procedure
we would complete searching for all simple loops in this term with no entries left. With such
statement, we can rewrite (\ref{eq:det-sum}) as summation over products of disjoint simple loops,
\begin{eqnarray}
  c_k = (-1)^k \sum_{\{j_1,\cdots,j_k\}} \sum
  \varepsilon_{i_1\cdots i_k} ([a]_{i_p\cdots i_p}) \cdots ([a]_{i_q\cdots i_q}) ,
  \label{eq:sum-over-loops}
\end{eqnarray}
where the second summation sums over all possible inequivalent simple loop combinations
whose lengths add up to $k$ and whose indices/vertices $i_l$ are chosen only from the set
$\{j_1,\cdots,j_k\}$.

Now let us focus on the second summation to see if we can extract more graphical
information therein. First of all we shall assume without loss of generality that the $k$
vertices have already been chosen to be $\{1,2,\cdots,k\}$. Second we would rewrite
the Levi-Civita symbol $\varepsilon$ as products of shorter $\varepsilon$'s each of which is associated with one simple loop. This can be most easily seen by a concrete example,
\begin{eqnarray}
  \varepsilon_{34521} a_{13} a_{24} a_{35} a_{42} a_{51}
  = \varepsilon_{351} \varepsilon_{42} (a_{13}a_{35}a_{51}) (a_{24}a_{42})
  = (\varepsilon_{351} [a]_{1351}) (\varepsilon_{42} [a]_{242}) .
  \label{eq:epsilon-split-example}
\end{eqnarray}
Generally, we have
\begin{eqnarray}
  \varepsilon_{i\cdots l} a_{1i} \cdots a_{kl} = (\varepsilon_{i\cdots} [a]_{1i\cdots 1})
  \cdots (\varepsilon_{j\cdots} [a]_{\cdots j\cdots}) .
  \label{eq:epsilon-split}
\end{eqnarray}
Then we can rewrite the second summation in (\ref{eq:sum-over-loops}) as
\begin{eqnarray}
  \sum_{\sum_i l_i=k} \prod_{i} \Big( (-1)^{l_i}
  \varepsilon_{j_i\cdots p_i} [a]_{\cdots j_i\cdots p_i\cdots} \Big) ,
  \label{eq:sum-over-loop-prod}
\end{eqnarray}
where $l_i$ is the length of the $i$-th simple loop in one combination of loops whose
lengths add up to $k$.

We want to prove that in the above product: (I) $(-1)^{l_i}\varepsilon_{j_i\cdots}=-1$,
(II) $[a]_{\cdots j_i\cdots}$ is the number of ways travelling through these vertices in
this simple loop. The second proposition can be readily understood by the fact that
$a_{ij}$ is the number of arrows from $i$ to $j$ which means there are $a_{ij}$ ways
to walk from $i$ to $j$, so the number of ways going through all vertices are their
product. The proof of the first proposition is slightly tricky and we need to take
advantage of symmetric groups.

Suppose we have $l$ vertices labeled with $\{j_1,j_2,\cdots,j_l\}$ where $j_1<j_2<\cdots<j_l$.
If there is a loop going through $j_1, j_2, \cdots, j_l$ successively, then there
will be a term in the product,
\begin{eqnarray}
  (-1)^l \varepsilon_{j_2j_3\cdots j_lj_1} a_{j_1j_2} a_{j_2j_3} \cdots a_{j_lj_1}
  = (-1)^l \varepsilon_{23\cdots l1} a_{j_1j_2} a_{j_2j_3} \cdots a_{j_lj_1}
  = (-1)^l \epsilon(\sigma) a_{j_1j_2} a_{j_2j_3} \cdots a_{j_lj_1} ,
  \label{eq:base-loop}
\end{eqnarray}
where $\sigma$ is an element of the symmetric group $S_l$,
\begin{eqnarray}
  \sigma = \left( \begin{array}{ccccc}
    1 & 2 & \cdots & l-1 & l \\
    2 & 3 & \cdots & l & 1
  \end{array} \right) .
  \label{eq:sigma}
\end{eqnarray}
On the right hand side of this equation we have used the signature (parity) function $\epsilon$ which
takes value 1 if the group element is an even permutation and takes value $-1$ otherwise.
Obviously $\epsilon(\sigma)=(-1)^{l-1}$ which gives an overall factor $-1$ in (\ref{eq:base-loop}).
For a generic loop going through $j_1$ to $j_l$ but not necessarily in an ascending order,
what is the overall factor? We have mentioned that equivalent loops differ only on which
vertex is the starting vertex. So we can safely set all inequivalent loops to start at
$j_1$, then any simple loop can be written as
\begin{eqnarray}
  (-1)^l (\varepsilon_{j_{i_2}\cdots}) a_{j_1j_{i_2}} a_{j_{i_2}j_{i_3}} \cdots a_{j_{i_l}j_1} .
  \label{eq:generic-loop}
\end{eqnarray}
Now the problem is how to determine other subscripts of $\varepsilon$. To determine the
$p$-th subscript we need to know which of these $a_{ij}$ has its first subscript being
$j_p$ and then the second subscript would be the $p$-th subscript of $\varepsilon$.
We first define another element of the same symmetric group,
\begin{eqnarray}
  \tau =  \left( \begin{array}{cccc}
    1 & 2 & \cdots & l \\
    1 & i_2 & \cdots & i_l
  \end{array} \right) ,
  \label{eq:tau}
\end{eqnarray}
and its inverse element
\begin{eqnarray}
  \tau^{-1} =  \left( \begin{array}{cccc}
    1 & i_2 & \cdots & i_l \\
    1 & 2 & \cdots & l
  \end{array} \right) .
  \label{eq:tau-inverse}
\end{eqnarray}
Through this permutation operator we can find out which $a_{ij}$ has its first subscript
equal to $j_p$, that is $j_p=j_{\tau(\tau^{-1}(p))}=j_{i_{\tau^{-1}(p)}}$. So the second
script would be $j_{i_{\tau^{-1}(p)+1}}=j_{i_{\sigma\tau^{-1}(p)}}=j_{\tau\sigma\tau^{-1}(p)}$.
Therefore the $\varepsilon$ in (\ref{eq:generic-loop}) can be explicitly calculated
\begin{eqnarray}
  \varepsilon_{j_{\tau\sigma\tau^{-1}(1)} \cdots j_{\tau\sigma\tau^{-1}(l)}}
  = \varepsilon_{\tau\sigma\tau^{-1}(1) \cdots \tau\sigma\tau^{-1}(l)}
  = \epsilon(\tau\sigma\tau^{-1}) = \epsilon(\sigma) = (-1)^{l-1} .
  \label{eq:explicit-epsilon}
\end{eqnarray}
Thus we proved our first proposition that $(-1)^l\varepsilon_{j\cdots}=-1$. Then
(\ref{eq:sum-over-loops}) simplifies to
$
  c_k = \sum_{\{j_1,\cdots,j_k\}} \sum_{\sum_i l_i=k} \prod_{a} (-[a])
$
and we come to conclusion in the Proposition \ref{thm:final-ck}.
\end{proof}

%%%%%%%%%%%%%%%%%%%%%%%%%%%%%%%%%%%%%%%%%%%%%%%%%%%%%%%%%%%%%%
\subsection{A Generic Formula for Superpotential}
From the aforementioned, we observed that the inverse of Ihara zeta function of fully directed quivers without self-adjoining loops can be explained by the number of disjoint simple cycles in the graph. In this setting, the inverse of Ihara zeta function is therefore a finite polynomial. More importantly, the inverse only looks at simple cycles in the quiver, which are simply terms in a generic superpotential. Therefore the Ihara zeta function, which is an infinite expansion of a finite polynomial, itself simply collects all possible combination of simple cycles of various lengths. 

While (\ref{eq:final-ck}) is the general conclusion, we would like to use this formula to calculate the reciprocal of Ihara zeta functions in a special case, where any two inequivalent loops share at least one point, so in (\ref{eq:final-ck}) there will be no
product since there is no disjoint loops according to this special setup. Therefore here we would have
\begin{eqnarray}
  c_k = - \sum_{{\rm length\ }k{\rm\ cycles}} [a] .
  \label{eq:special-ck}
\end{eqnarray}
That is, $c_k$ is the negative of the number of all inequivalent simple cycles of length $k$ in the corresponding graph. 
Note that traversing the same vertex sequence but through different arrows is considered as traversing different cycles.

In this perspective, the Ihara zeta function gives the most generic superpotential when given information on the quiver alone (of course, in obtaining Calabi-Yau moduli spaces for certain theories one needs specific couplings). 
We now introduce a {\bf refinement} of the Ihara zeta function, \emph{i.e}, coefficients of the monomials in expansion will be more than just one complex variable $z$ but one for each arrow in the quiver.
To give some illustrative examples, let us consider the $dP_0$ and $F_0$ (cones over $\mathbb{P}^2$ and $\mathbb{P}_1 \times \mathbb{P}_1$) cases separately. 

For $dP_0$, we have 9 fields $\Psi_{1, \ldots ,9}$ and we shall replace the entries in the adjacency matrix with these fields:
\begin{equation}
Q_{\mathbb{P}^2}' = \left(
\begin{array}{ccc}
0 & \Psi_1 + \Psi_2 + \Psi_3 & 0 \\
0 & 0 & \Psi_4 + \Psi_5 + \Psi_6\\
\Psi_7 + \Psi_8 + \Psi_9 & 0 & 0
\end{array}
\right) \ . 
\end{equation}
So we have the inverse of Ihara zeta function as 
\begin{equation}
\begin{aligned}
\zeta_{Q_{\mathbb{P}^2}'}^{-1} & = 1 - (\Psi_1 + \Psi_2 + \Psi_3)(\Psi_4 + \Psi_5 + \Psi_6)(\Psi_7 + \Psi_8 + \Psi_9)z^3 \\
 & = 1 - (\Psi_1\Psi_4\Psi_7 + \Psi_2\Psi_4\Psi_7 + \Psi_3\Psi_4\Psi_7 + \Psi_1\Psi_5\Psi_7 + \Psi_2\Psi_5\Psi_7 + \Psi_3\Psi_5\Psi_7\\
  & + \Psi_1\Psi_6\Psi_7 + \Psi_2\Psi_6\Psi_7 + \Psi_3\Psi_6\Psi_7 + \Psi_1\Psi_4\Psi_8 + \Psi_2\Psi_4\Psi_8 + \Psi_3\Psi_4\Psi_8\\
  & + \Psi_1\Psi_5\Psi_8 + \Psi_2\Psi_5\Psi_8 + \Psi_3\Psi_5\Psi_8 + \Psi_1\Psi_6\Psi_8 + \Psi_2\Psi_6\Psi_8 + \Psi_3\Psi_6\Psi_8\\ 
  & + \Psi_1\Psi_4\Psi_9 + \Psi_2\Psi_4\Psi_9 + \Psi_3\Psi_4\Psi_9 + \Psi_1\Psi_5\Psi_9 + \Psi_2\Psi_5\Psi_9 + \Psi_3\Psi_5\Psi_9\\ 
  & + \Psi_1\Psi_6\Psi_9 + \Psi_2\Psi_6\Psi_9 + \Psi_3\Psi_6\Psi_9)z^3.
\end{aligned}
\end{equation}
Here we have all possible cubic superpotential terms that can be constructed from the quiver, however the true full superpotential for $dP_0$ theory is (e.g., see the catalogue in \cite{Davey:2009})
\begin{equation}
\mathcal{W}_{\mathbb{P}^2} = \Psi_1\Psi_4\Psi_7 - \Psi_1\Psi_6\Psi_8 - \Psi_2\Psi_4\Psi_9 +
\Psi_2\Psi_5\Psi_8 - \Psi_3\Psi_5\Psi_7 + \Psi_3\Psi_6\Psi_9.
\end{equation}

Therefore, the expansion of Ihara zeta function in terms of bi-fundamental fields will give the most {\em generic} superpotential of a specific quiver gauge theory, where the constant 1 in the expansion should be removed as well as changing the signs of non-constant in $\zeta^{-1}$.
Moreover, we can also interpret $z^3$ as the coupling for all the terms in superpotential. The moduli space here would be trivially a point since the F-terms have as many non-trivial constraints as there are number of fields. In other words, the Jacobian ideal of the generic superpotential is trivial.

Next, let us consider a Seiberg dual phase for the $F_0$ theory. We have its adjacency matrix in the following form:
\begin{equation}
Q_{\mathbb{P}^1 \times \mathbb{P}^1}' = 
\left(\begin{array}{cccc}
0 & 0 & 0 & \sum_{n=1}^6 \Psi_n \\
\sum_{n=7}^8 \Psi_n & 0 & 0 & 0 \\
0 & \sum_{n=9}^{10} \Psi_n & 0 & 0 \\
0 & \sum_{n=11}^{18} \Psi_n & \sum_{n=19}^{24} \Psi_n & 0
\end{array} \right) \ .
\end{equation}
With this adjacency matrix, we have its Ihara zeta function to be
\begin{equation}
\zeta_{Q_{\mathbb{P}^1 \times \mathbb{P}^1}'}^{-1} = 
1 - \sum_{n=1}^6 \Psi_n \sum_{n=7}^8 \Psi_n (z \sum_{n=11}^{18} \Psi_n + z^2 \sum_{n=9}^{10} \Psi_n \sum_{n=19}^{24} \Psi_n) z^2,
\end{equation}
where it is obvious that all possible combinations of cubic and quartic terms are generated, and upon inverting, we have an infinite polynomial generated from the above finite one.

%%%%%%%%%%%%%%%%%%%%%
\subsection{Adjoint Fields and Bi-Directional Arrows}
\label{sec:SPP}

While most of our $\mathcal{N}=1$ gauge theories consist only of directed arrows, there are some that have self-adjoint loops (adjoint fields) as well as bi-directional arrows, such as the SPP theory and the dual phases of $L^{aba}$ theories. For these, the form of the Ihara zeta function are more complicated since we can replace each pair of directed arrows between nodes by a single undirected edge
and it would be interesting to see the difference between treating these edges as directed
and treating them as undirected. In the following, we will take SPP case as an example to
illustrate this difference.

The SPP quiver \cite{Morrison:1998, Franco:2005} is shown in Figure \ref{fig:SPP} from
which we can easily read off its adjacency matrix,
\begin{eqnarray}
  A = \left( \begin{array}{ccc}
    2 & 1 & 1 \\
    1 & 0 & 1 \\
    1 & 1 & 0
  \end{array} \right) .
  \label{eq:SPP-adj}
\end{eqnarray}
\begin{figure}[!ht]
\centering
	\begin{tikzpicture}[->,>=stealth',shorten >= 1pt,auto,semithick,
	main node/.style={circle,fill=green},
	inverse/.style={<-,shorten <= 1pt, auto, semithick}]
		\pgfmathsetmacro{\y}{4*sin(60)}
		\pgfmathsetmacro{\c}{4*sin(60)+.5}
		\pgfmathsetmacro{\r}{.5}
        \draw [->] (-.707*\r,\c-.707*\r) to [out=135,in=225] (-.707*\r,\c+.707*\r)
              to [out=45,in=135] (.707*\r,\c+.707*\r) to [out=315,in=45] (.707*\r,\c-.707*\r);
		\node[main node]		(1)		at (0,\y)		{1};
		\node[main node]		(2)  	at (-2,0)		{2};
		\node[main node]		(3)  	at (2,0)		{3};
		\path 	(2)		edge		[<->]		node		{}		(1)
				(3)		edge		[<->]		node		{}		(2)
				(1)		edge		[<->]		node		{}		(3);
	\end{tikzpicture}
    \caption{\sf SPP Quiver. All arrows are bi-directional and graphically they can be replaced by undirected edges.}
	\label{fig:SPP}
\end{figure}
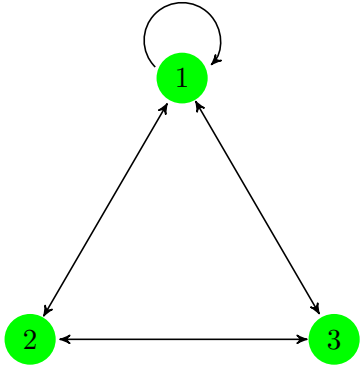
There are several different ways to view its edges:
\begin{itemize}
  \item[1] All edges are considered to be directed, thus a fully directed graph. We then can still use
    \eqref{eq:fully-directed-zeta} to compute Ihara zeta function.
  \item[2] All edges are considered to be undirected, then it is an undirected graph. We then need to
    utilize an undirected version of formula \eqref{eq:partial_zeta},
    \begin{eqnarray}
      \zeta(z) &=& \frac{(1-z^{2})^{-\rm{Tr}(Q-I)/2}}{{\rm Det}(I-Az+Qz^{2})},
      \label{eq:undirected-zeta}
    \end{eqnarray}
    to compute the zeta function.
  \item[3] Only the self-adjoint loop is considered as undirected, thence we get a partially directed graph.
    In this case we have to use Eq. \eqref{eq:partial_zeta} itself to compute Ihara zeta function.
\end{itemize}

In the first perspective, the Ihara zeta function is extremely simple,
\begin{eqnarray}
  \zeta^{-1}(z) = (1+z)(1-3z) .
  \label{eq:SPP-fully-directed}
\end{eqnarray}
It can be readily seen that there are only two poles $z=-1$ and $z=1/3$ both being real.
For the above second case, we have $Q=\textrm{diag}\{3,1,1\}$ and the zeta function
follows immediately,
\begin{eqnarray}
  \zeta^{-1}(z) &=& (1-z)^2(1+z)(1+z+z^2)(1-2z+2z^2-3z^3) .
  \label{eq:SPP-undirected}
\end{eqnarray}
The poles of this function are plotted in Figure \ref{fig:SPP-poles} as green and blue points.
\begin{figure}[!t]
  \centering
  \includegraphics[width=.5\textwidth]{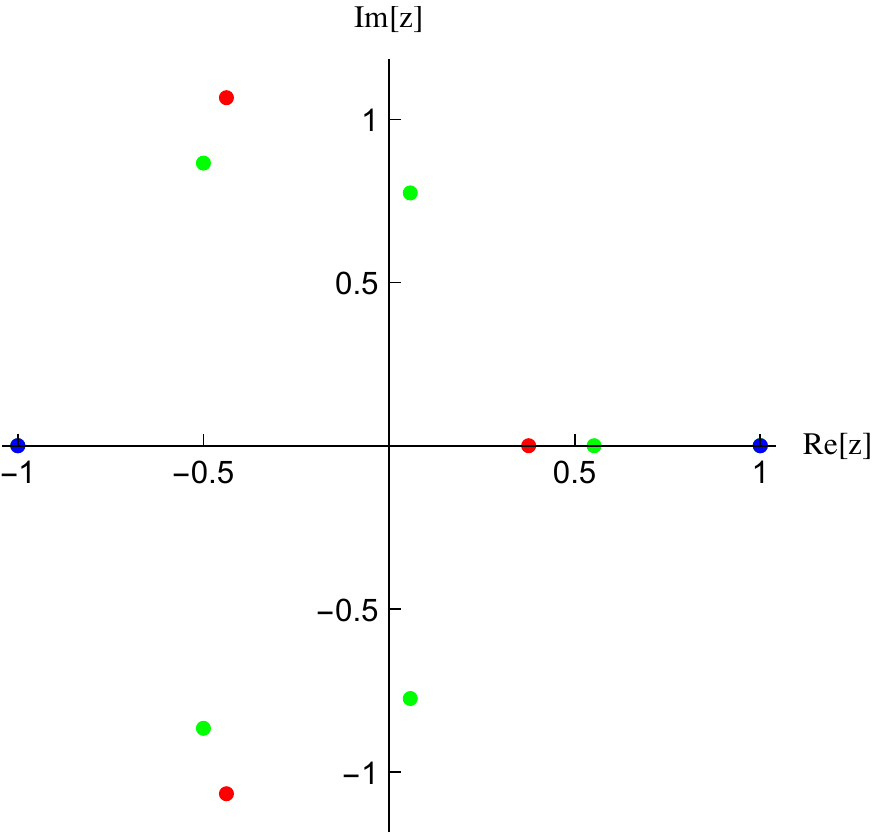}
  \caption{\sf Poles of Ihara zeta function for the SPP quiver. The green and blue points are poles
  when we see all the edges as undirected ones, whereas the red and blue points are poles when
  we take SPP quiver as a partially directed graph.}
  \label{fig:SPP-poles}
\end{figure}
In the last case, $Q$ only
counts undirected degrees and $P$ is only responsible for directed edges, so we find
$Q=\textrm{diag}\{1,-1,-1\}$ and
\begin{eqnarray}
  P = \left( \begin{array}{ccc}
    0 & 1 & 1 \\
    1 & 0 & 1 \\
    1 & 1 & 0
  \end{array} \right) .
  \label{eq:SPP-P-matrix}
\end{eqnarray}
Then the resulting zeta function reads
\begin{eqnarray}
  \zeta^{-1}(z) = (1-z)(1+z)(1-2z-z^2-2z^3) .
  \label{eq:SPP-partial}
\end{eqnarray}
The poles in this case are also plotted in Figure \ref{fig:SPP-poles} as red and blue points.

From the above example, one can see that the behaviors of poles for Ihara zeta function
are quite different when we treat a graph with self-adjoint loops and bi-directional
arrows from different perspectives.

Next we turn to the refined Ihara zeta function in which the adjacency matrix $A$ takes
fields as its entries other than just the number of arrows/edges/loops.
For the first case viewing all edges as arrows, we have
\[ A = \left( \begin{array}{ccc}
    2\phi & \Psi_1 & \Psi_2 \\
    \Psi_3 & 0 & \Psi_4 \\
    \Psi_5 & \Psi_6 & 0
\end{array} \right), \]
which results in the Ihara zeta function
\begin{eqnarray}
  \zeta^{-1}_{SPP}(z) = 1 - 2 \phi z - (\Psi_1\Psi_3 + \Psi_2\Psi_5 + \Psi_4\Psi_6) z^2
 - (\Psi_1\Psi_4\Psi_5 + \Psi_2\Psi_3\Psi_6 - 2 \phi\Psi_4\Psi_6) z^3 .
 \label{eq:zeta-SPP}
\end{eqnarray}
Therefore we still have the refined Ihara zeta function as the generating function
of generic superpotentials. However, for the second and third cases, we have to face the
problem of how to define the $Q$ and $P$ matrices. Since $Q$ is in the power, we would
assume it remains to be a matrix of numbers, whereas $P$ can be assumed to comprise of either numbers
or fields. Nonetheless, in either assumption the final Ihara zeta function cannot be
thought of as a superpotential generating function.

%%%%%%%%%%%%%%%%%%%%%%%%%%%%%%%
%%%
%%%%%%%%%%%%%%%%%%%%%%%%%%%%%%%
\section{Conclusions and Prospects}\label{s:conc}
In this paper, we have investigated the Ihara zeta function associated to quiver gauge theories, especially under the action of Seiberg duality.
Aid by a multitude of explicit examples, we have learnt many lessons.

In the $\mathbb{P}^2$ case, the Seiberg dual theories are well-known to be in one-to-one correspondence to solutions of the Markov Diophantine equation and
we proved that all poles of the Ihara zeta functions are on the lines of cubic roots of unity.
For $\mathbb{P}^1 \times \mathbb{P}^1$ theories, there is also an underlying Diophantine equation of a so-called 3-block structure \cite{Benvenuti:2004dw,Hanany:2012mb}. Due to the complexity here, we adopted a numerical approach and found that every (even) solution with its components being even and less than 800 corresponds to some quiver theory that is a Seiberg dual phase.

Inspired by the graph version of the Riemann Hypothesis, We studied the distribution of poles of the Ihara zeta function along the Seiberg duality tree of the various theories, including the various del Pezzo and Hirzebruch geometries.
The results are shown in figure \ref{fig:pole_plots} where we can see that in all cases, the poles are {\em concentrated in strip areas} on or near the coordinate axes, which shows an axial symmetry about the horizontal axis.
While there is yet to be a formulation for the Riemann Hypothesis for irregular digraphs as in our quivers, it is still reasonable to speculate that all poles should be constrained within some area other than randomly distributed everywhere in the complex plane.

Nevertheless, we can still use the definitions of strong and weak versions of the graph Riemann Hypothesis (cf.~Eqs.\eqref{def:strong_RH} and \eqref{def:weak_RH}) with certain choice of (max) degree $q$ to see how it acts on Seiberg duals. Numerical results are summarized in \S\ref{s:numerical_riemann}. It is interesting to observe how this RH setup selects certain quivers in the duality tree and hence produce specific prime sequences in the coefficients of the zeta function.

As a interesting by-product, we find a graph-theoretic interpretation for the reciprocal of Ihara zeta function which is summarized in \eqref{eq:final-ck} and which has gauge theoretic repercussions. As a consequence of this formula, it is immediate to conclude that a {\em refined version} of the Ihara zeta function is a generating function for the generic superpotential (with couplings in powers of $z$) for the quiver gauge theory.

One of the immediate challenges is to have a reasonable generalization of the graph Riemann Hypothesis for directed quivers. While there is a lack of functional equation in such a case, we have seen in our cases that the distribution of the poles is not random and that under Seiberg duality, they remain in constrained regions.

In section \ref{sec:SPP}, we have seen that the refined Ihara zeta function cannot
be a generating function of the SPP quiver. So another interesting work is to find
a closed form (may not be a Ihara zeta type function) for generic superpotentials
of undirected quivers or partially directed quivers comprised of directed edges and
undirected self-loops. After all, these two kinds of quivers both have significant
meanings in quiver gauge theories.

%%%%%%%%%%%%%%%%%%%%%%%%%%%%%%%%%%%%%%%%%%%%%%%%%%%%%
\section*{Acknowledgements}
We thank USTC Cloud Lab for providing us with a 16-core virtual machine for some intensive computations of this paper. 
YHH is indebted to the Science and Technology Facilities Council, UK, for grant ST/J00037X/1, the Chinese Ministry of Education, for a Chang-Jiang Chair Professorship at NanKai University, and the city of Tian-Jin for a Qian-Ren Award.
YHH is also indebted to Merton College, Oxford for support.
DZ acknowledges the support from the National Science Foundation of China under grant No. 11105138 and 11235010. He is also supported by the Chinese Scholarship Council (CSC).
YX acknowledges the support from the Doctoral Scholarship of City University London.

\appendix
\appendixpage
\section{Ihara Zeta Function: Determinants and Euler-Product}
\label{ap:euler_determinant}
In \cite{Bass}, another zeta function called \textbf{edge zeta function} is constructed to give a determinant form of Ihara zeta function for undirected graph. Before giving the full definition of this zeta function, we need a few other definitions:
\begin{itemize}
	\item The \textbf{edge matrix} $W$ of size $2m \times 2m$ for an undirected graph with $m$ undirected edges has entries $w_{ij}$. The $(i,j)$-th entry of $W$, $w_{ij}$, is a complex variable if the edge $e_i$ is connected with edge $e_j$ with $e_j \neq e_i^{-1}$, and the entry is 0 if otherwise.
	\item For a closed path $C$ in an undirected graph $X$ written as a sequence of edges $C = e_1e_2 \cdots e_s$, the \textbf{edge norm} of $C$ is
	\[N_E(C) = w_{12}w_{23} \cdots w_{s1}.\]
\end{itemize}
With the above definitions in hand, we can define the edge zeta function as follows
\[\zeta_E(W,X) = \prod_{[P] \in \rm{Prime \ Cycles}} (1-N_E(C))^{-1}. \]
It is clear from this definition that if $w_{ij}$ is set to $z \in \mathbb{C}$, we recover the original Ihara zeta function such that
\[\zeta_G(z) = \zeta_E(W_1,G),\]
where $W_1$ is the edge matrix when all non-zero entries set to $z$.

Furthermore, we have the following theorem (cf. Chapter 3 of \cite{Terras})
\begin{theorem}
\[\zeta_E(W,G) = {\rm Det}(I - W)^{-1}.\]
\label{thm:edge_zeta}
\end{theorem}
\begin{proof}
By taking logarithm of Euler-product, we have
\[-{\rm log}\zeta_E(W,G) = \sum_{[P]} \sum_{j \geq 1} \frac{1}{j} N_E(P)^j,\]
where Taylor expansion of logarithm is used. As we have $l(P)$ elements in the prime equivalent class $[P]$, hence
\[-{\rm log}\zeta_E(W,G) = \sum_{\substack{m \geq 1 \\ j \geq 1}}\frac{1}{jm} \sum_{\substack{P \\ l[P] = m}} N_E(P)^j.\]
The sum now is over all prime paths. This is equivalent to sum over all paths of the form $C = P^j$ with length $jm$. It then follows that
\[-{\rm log}\zeta_E(W,G) = \sum_C \frac{1}{l[C]} N_E(C).\]
Since we also gave the following
\[ \sum_C \frac{1}{l[C]} N_E(C) = \sum_{m \geq 1} \frac{1}{m} {\rm Tr}(W^m) = -{\rm Trlog}(I-W) = -{\rm logDet}(I-W)^{-1},\]
where the first equality can be understood as the trace over $W^m$ is collecting all cycles of various length, then we have $\zeta_E(W,G) = {\rm Det}(I-W)^{-1}.$
\end{proof}
Before we give proof to the determinant formula for Ihara zeta function, some matrix identities are needed. Proofs of these identities are given in Chapter 3 of \cite{Terras}.
The following definitions are needed in the proof:

First set $J = \left( \begin{array}{cc}
0 & I_m\\
I_m & 0
\end{array}
\right)$. Then the $n \times 2m$ start matrix $S$ and the $n \times 2m$ terminal matrix $T$ are defined as
\[ S_{ve} = 
\begin{cases}
1, & \text{if $v$ is the starting vertex of edge $e$},\\
0, & \text{otherwise,}
\end{cases}\]

and
\[ T_{ve} = 
\begin{cases}
1, & \text{if $v$ is the terminal vertex of edge $e$},\\
0, & \text{otherwise.}
\end{cases}\]
With the above definitions, we have the following identities
\begin{enumerate}
	\item $SJ = T$ and $TJ = S$.
	\item If $A$ is the adjacency matrix of graph $G$, and $Q+I_n$ is the undirected degree matrix whose $j$-th diagonal entry specifies the degree of $j$-th vertex in $G$. Then we have $A=ST^T$ and $Q+I_n = SS^T = TT^T$.
	\item The edge matrix $W_1$ obtained by setting all non-zero entries of $W$ to $z$ has the identity $W_1 + J = T^TS$. Here $M^T$ is the usual notation for matrix transpose.
\end{enumerate}

Now we come to the proof of determinant formula of Ihara zeta function for undirected graphs, \emph{i.e.}, $\zeta_G(z) = (1-z^2)^{-{\rm Tr}(Q-I)/2}{\rm Det}(I - Az + Qz^2)^{-1}$. This proof is elaborated in more details in \cite{Bass}. First consider a graph with $m$ undirected edges and $n$ vertices, we have following matrix equation from previous identities
\[ \left(
\begin{array}{cc}
I_n & 0 \\
T^T & I_{2m}
\end{array}
\right) \left(
\begin{array}{cc}
I_n(1-z^2) & Sz \\
0 & I_{2m} - W_1z
\end{array}
\right)  = 
\left(
\begin{array}{cc}
I_n - Az + Qz^2 & Sz \\
0 & I_{2m} + Jz
\end{array}
\right)
\left(
\begin{array}{cc}
I_n & 0 \\
T^T - S^Tz & I_{2m}
\end{array}
\right).\]
In the above equation, all matrices are of $(n+2m) \times (n+2m)$, and we take determinant on both sides of equation to give
\[(1-z^2)^n {\rm Det}(I-W_1z) = {\rm Det}(I_n - Az + Qz^2){\rm Det}(I_{2m} + Jz).\]
With the observation that
\[I+Jz = \left(
\begin{array}{cc}
I & Iz\\
Iz & I
\end{array}
\right)\] implies
\[\left(
\begin{array}{cc}
I & 0\\
-Iz & I
\end{array}
\right) (I + Iz) = \left(
\begin{array}{cc}
I & Iz\\
0 & I(1-z^2)
\end{array}
\right),\]
we obtain ${\rm Det}(I + Jz) = (1 - z^2)^m$. Therefore it is clear that
\[{\rm Det}(I - W_1z) = {\rm Det}(I_n - Az + Qz^2) {\rm Det} (1-z^2)^{m-n},\]
where the LHS is now $\zeta_G(z)^{-1}$ from theorem \ref{thm:edge_zeta}. Since $m-n = {\rm Tr}(Q-I)/2$, we have the determinant formula of Ihara zeta function.

For the details of proof of \eqref{eq:partial_zeta} that gives determinant expression of Ihara zeta function for partially directed graphs, see \cite{Tarfuleaa}.

%%%%%%%%%%%%%%%%%%%%%%%%%%%%%%%%%%%%%%%%%%%%%%%%%%%%%
\section{Solutions to the Markov Equation}
\label{ap:markov}

Most of the following proofs can be found in the book \cite{cassels:1965}, we
nonetheless rephrase them here for the purpose of integrality of our logic flow
in Section \ref{sec:p2_pole}. However, in order to conform with our underlying
physics context, instead of using the original definition of Markov equation
\begin{eqnarray}
  u^2 + v^2 + w^2 = 3 uvw \ ,
  \label{eq:def-original-markov}
\end{eqnarray}
we define it here as
\begin{eqnarray}
  u^2 + v^2 + w^2 = uvw \ ,
  \label{eq:def-markov}
\end{eqnarray}
which is essentially the same. Since in physics we are only interested in
positive solutions, the positivity of solutions are always presumed in
following deductions.

\begin{lemma}
  If $(a,b,c)$ is a solution to eq.~(\ref{eq:def-markov}), $(bc-a,b,c)$
  is also a solution as well as $(a,ac-b,c)$ and $(a,b,ab-c)$.
  \label{thm:another-solution}
\end{lemma}

\begin{proof}
It is readily seen that eq.~(\ref{eq:def-markov}) is symmetric under permutation
of variables $(u,v,w)$, it is therefore sufficient to prove one of these combinations
satisfies eq.~(\ref{eq:def-markov}).
\begin{displaymath}
  (bc - a)^2 + b^2 + c^2 = (bc - a)bc + (a^2 + b^2 + c^2 - abc) = (bc - a)bc .
\end{displaymath}
\end{proof}

Comparing this with the Seiberg duality rules in Section~\ref{sec:seiberg-dual},
one would see that $(a,b,c)\rightarrow(bc-a,b,c)$ is exactly of Seiberg duality
type apart from a common factor $-1$. Thus we shall call this operation on solutions
``\emph{Seiberg transformation}''.
There are two obvious solutions, $(3,3,3)$ and $(3,3,6)$ (of course, $(3,6,3)$
and $(6,3,3)$ are considered as the same solutions as $(3,3,6)$ due to the
symmetry under permutation). 
These two are dubbed \emph{singular solutions} because
there are equal numbers in both solutions. \emph{Non-singular solutions} are then
defined as solutions with three distinct numbers.

\begin{lemma}
  All solutions to Markov equation except the above two are non-singular.
  \label{thm:non-singular}
\end{lemma}

\begin{proof}
  Suppose we have a singular solution $(a,b,c)$, without loss of generality,
  we shall assume $a=b$, then $2a^2+c^2=a^2c\ \Rightarrow\ c^2=(c-2)a^2$.
  This indicates $\sqrt{c-2}=c/a$ is a rational number, which means $c-2$
  has to be a perfect square, say $d^2$ where $d$ is an integer. Hence with
  $c=ad$, we have $2a^2+a^2d^2=da^3\ \Rightarrow\ 2=d(a-d)$, which means
  $d=1$ or 2. In either case we will have $a=3$. Then $c=ad=3$ or 6. Thus
  the two solutions $(3,3,3)$ and $(3,3,6)$. Any other solutions would have
  $u$, $v$, $w$ being distinct.
\end{proof}

Now for all non-singular solutions, we can write them in a manner such that
each has its components in ascending order --- $(u,v,w)$ with $u<v<w$.

\begin{lemma}
  If we have a non-singular solution $(a,b,c)$ with $a<b<c$,
  then $ab-c<b$ and $bc-a>ac-b>c$.
  \label{thm:adjacent-solution-order}
\end{lemma}

\begin{proof}
Consider the quadratic function
\begin{eqnarray}
  f(x) = x^2 - abx + a^2 + b^2 \ ,
  \label{eq:quadratic-f}
\end{eqnarray}
which has two zero points, $c$ and $ab-c$. Also we can evaluate $f(x)$ at $x=b$,
\begin{eqnarray}
  f(b) = a^2 - (a-2)b^2 \leq a^2 - b^2 < 0 \ ,
  \label{eq:f-b}
\end{eqnarray}
where in the second step we have used the fact that $a\geq 3$. Now that
$f(x)$ is an upwards concave parabola, what immediately follows is that
$ab-c<b<c$. With the same trick we can prove $ac-b>c$ and by definition
we have $bc-a>ac-b$.
\end{proof}

Now assuming $a<b<c$, we shall call $(a,b,c)\rightarrow(ab-c,a,b)$
\emph{descending transformation} and the other two \emph{ascending transformation}.
The above lemma tells us that by Seiberg transforming any non-singular
solution, we would end up with a ``bigger'' or ``smaller'' solution
depending on which node (or which arrows) your are dualizing on.
As for singular solutions, if we Seiberg transform $(3,3,6)$ we would
get $(3,3,3)$ (which is smaller) or a non-singular solution (which is
bigger), and if we Seiberg transform $(3,3,3)$ we can only get $(3,3,6)$.
Following this logic, if we we have a non-singular solution, then by
repeatedly apply descending transformation, we will constantly get
smaller solutions until there are no smaller ones, that is, it will
stop at $(3,3,3)$. This also means that we can reverse all transformation steps
to generate any non-singular solution from $(3,3,3)$. 
Thus we have:
\begin{theorem}
All positive solutions of Markov equation can be generated from $(3,3,3)$
by Seiberg transformation.
\end{theorem}
The solutions generated from a given solution are called \emph{adjacent solutions}.
For example, $(3yz-x,y,z)$ is called the adjacent solution to $(x,y,z)$ with respect to $x$.


\begin{thebibliography}{99}

%%%%%%%%%%%%%%%%%%%%% Graph Zeta


\bibitem{ihara}
Y.~Ihara, ``On discrete subgroups of the two by two projective linear group over p-adic fields". J. Math. Soc. Japan 18 (1966) 219-235.

\bibitem{Horton}
M.~Horton, ``Zeta functions of digraphs'',
Linear Algebra Appl. 425 (1) (2007) 130-142.

\bibitem{M-S}
H.~Mizuno and I.~Sato, ``Zeta functions of digraphs'', 
Linear Algebra Appl. 336 (2001) 181-190.

\bibitem{Bass}
H.Bass, ``The Ihara-Selberg zeta function of a tree lattice",
Internat.J.Math.,3(1992),717-797.

\bibitem{Terras}
Audrey Terras, ``Zeta Functions of Graphs: A Stroll through the Garden'',
Cambridge Studies in Advanced Mathematics (No. 128).
\\
cf.~\verb|http://math.ucsd.edu/~aterras/|

\bibitem{Sunada:LF}
L-functions in geometry and some applications, 
Proc. Taniguchi Symp. 1985,``Curvature and Topology of Riemannian Manifolds", Springer Lect. Note in Math. 1201(1986), 266-284.

\bibitem{Sunada:FG}
Fundamental groups and Laplacians, 
Proc. Taniguchi Symp.``Geometry and Analysis on Manifolds", 1987, Springer Lect. Note in Math. 1339(1988), 248-277.

\bibitem{Tarfuleaa} Andrei Tarfuleaa, and Robert Perlisb,
	``An Ihara formula for partially directed graphs”,
	Linear Algebra and its Applications, Volume 431, Issues 1-2, 1 July 2009, Pages 73-85.

\bibitem{Murty}
R.Murty,
``Ramanujan Graphs",
Journal of the Ramanujan Math.Society, 18, No.1. 2007.

\bibitem{Winnie:NT}
W.C. Winnie Li,
``Number Theory with Applications",
Series of University Mathematics, Vol. 7, World Scientific, 1996.

\bibitem{Winnie:RD}
W.C. Winnie Li,
``Recent developments in automorphic forms and applications",
Number Theory for the Millenium, Volume 2, edited by M.A. Bennett, B.C. Berndt, N. Boston, H.G. Diamond, A.J. Hildebrand, and W. Philipp, 2002, pp. 331–354, A.K. Peters,
Natick, Massachusetts
\bibitem{Lubotzky}A. Lubotzky,
``Discrete Graphs, Expanding Graphs and Invariant Measures",
Progress in Mathematics, 125 Birkhauser, 1994.

\bibitem{Valette}
A. Valette,
``Graphes de Ramanujan, Asterisque",
245 (1997), 247–276.

\bibitem{mckay}
J.~McKay, ``Graphs, singularities and finite groups'', 
Proc. Symp. Pure Math. AMS 37: 183\\
J.~McKay, ``Cartan matrices, finite groups of quaternions, and Kleinian singularities'', Proc.~AMS, 81, 153 - 154.

\bibitem{He:1999xj} 
  Y.~H.~He,
  ``Some remarks on the finitude of quiver theories,''
  [hep-th/9911114].
  %%CITATION = HEP-TH/9911114;%%

\bibitem{He:2015jla} 
  Y.~H.~He, V.~Jejjala and D.~Minic,
  ``From Veneziano to Riemann: A String Theory Statement of the Riemann Hypothesis,''
  arXiv:1501.01975 [hep-th].
  %%CITATION = ARXIV:1501.01975;%%

\bibitem{chen:1997}
  Wai-Kai Chen,
  Graph Theory and its Engineering Applications, Vol. 5,
  River Edge, New Jersey: World Scientific, 1997.


%%%%%%%%%%%%%%%%%%%%%%%%%%%%%%%%%GAUGE THEORY

\bibitem{He:2011ge} 
  Y.~H.~He,
  ``Graph Zeta Function and Gauge Theories,''
  JHEP {\bf 1103}, 064 (2011)
  [arXiv:1102.1304 [math-ph]].
  %%CITATION = ARXIV:1102.1304;%%


\bibitem{2} N. Seiberg,
	``Electric - magnetic duality in supersymmetric nonAbelian gauge theories,” Nucl. Phys. B 435, 129 (1995) [arXiv:hep-th/9411149].

\bibitem{3} K. A. Intriligator, R. G. Leigh and M. J. Strassler,
	``New examples of duality in chiral and nonchiral supersymmetric gauge theories,”
	Nucl. Phys. B 456, 567 (1995) [arXiv:hep-th/9506148].

\bibitem{4} P. S. Aspinwall and A. E. Lawrence,
	``Derived categories and zero-brane stability,”
	JHEP 0108, 004 (2001) [arXiv:hep-th/0104147].

\bibitem{5} A. Giveon and D. Kutasov,
	``Brane dynamics and gauge theory,”
	Rev. Mod. Phys. 71, 983 (1999) [arXiv:hep-th/9802067].

\bibitem{6} A. Hanany and E. Witten,
	``Type IIB superstrings, BPS monopoles, and three-dimensional gauge dynamics,”
	Nucl. Phys. B 492, 152 (1997) [arXiv:hep-th/9611230].

\bibitem{7} S. Elitzur, A. Giveon and D. Kutasov,
	``Branes and N = 1 duality in string theory,”
	Phys. Lett. B 400, 269 (1997) [arXiv:hep-th/9702014].


%%%%%%%%%%%%% quiver duality
\bibitem{Feng:2000mi} 
  B.~Feng, A.~Hanany and Y.~H.~He,
  ``D-brane gauge theories from toric singularities and toric duality,''
  Nucl.\ Phys.\ B {\bf 595}, 165 (2001)
  [hep-th/0003085].
  %%CITATION = HEP-TH/0003085;%%

\bibitem{Feng:2001bn} 
  B.~Feng, A.~Hanany, Y.~H.~He and A.~M.~Uranga,
  ``Toric duality as Seiberg duality and brane diamonds,''
  JHEP {\bf 0112}, 035 (2001)
  [hep-th/0109063].
  %%CITATION = HEP-TH/0109063;%%

\bibitem{Cachazo:2001gh} 
  F.~Cachazo, S.~Katz and C.~Vafa,
  ``Geometric transitions and N=1 quiver theories,''
  hep-th/0108120.
  %%CITATION = HEP-TH/0108120;%%



\bibitem{Beasley:2001zp} 
  C.~E.~Beasley and M.~R.~Plesser,
  ``Toric duality is Seiberg duality,''
  JHEP {\bf 0112}, 001 (2001)
  [hep-th/0109053].
  %%CITATION = HEP-TH/0109053;%%

\bibitem{FZ}
Sergey Fomin, Andrei Zelevinsky
``Cluster algebras I: Foundations'', arXiv:math/0104151 [math.RT]


\bibitem{Cecotti:1992rm} 
  S.~Cecotti and C.~Vafa,
  ``On classification of N=2 supersymmetric theories,''
  Commun.\ Math.\ Phys.\  {\bf 158}, 569 (1993)
  [hep-th/9211097].
  %%CITATION = HEP-TH/9211097;%%


\bibitem{Feng:2002kk} 
  B.~Feng, A.~Hanany, Y.~H.~He and A.~Iqbal,
  ``Quiver theories, soliton spectra and Picard-Lefschetz transformations,''
  JHEP {\bf 0302}, 056 (2003)
  [hep-th/0206152].
  %%CITATION = HEP-TH/0206152;%%


\bibitem{Franco:2003ja} 
  S.~Franco, A.~Hanany, Y.~H.~He and P.~Kazakopoulos,
  ``Duality walls, duality trees and fractional branes,''
  hep-th/0306092.
  %%CITATION = HEP-TH/0306092;%%

\bibitem{Franco:2003ea} 
  S.~Franco, A.~Hanany and Y.~H.~He,
  ``A Trio of dualities: Walls, trees and cascades,''
  Fortsch.\ Phys.\  {\bf 52}, 540 (2004)
  [hep-th/0312222].
  %%CITATION = HEP-TH/0312222;%%


\bibitem{Benvenuti:2004dw} 
  S.~Benvenuti and A.~Hanany,
  ``New results on superconformal quivers,''
  JHEP {\bf 0604}, 032 (2006)
  [hep-th/0411262].
  %%CITATION = HEP-TH/0411262;%%


\bibitem{Hanany:2012mb} 
  A.~Hanany, Y.~H.~He, C.~Sun and S.~Sypsas,
  ``Superconformal Block Quivers, Duality Trees and Diophantine Equations,''
  JHEP {\bf 1311}, 017 (2013)
  [arXiv:1211.6111 [hep-th]].
  %%CITATION = ARXIV:1211.6111;%%

\bibitem{12}
  Li, W.-C. W.,
  ``Character sums and abelian Ramanujan graphs",
  J. Number Theory, 41 (1992), 199-217.

\bibitem{cassels:1965}
  J.~W.~S.~Cassels,
  ``An introduction to Diophantine approximation'',
  Cambridge University Press (1965).

\bibitem{Davey:2009} 
  J.~Davey, A.~Hanany and J.~Pasukonis,
  ``On the Classification of Brane Tilings,''
  JHEP {\bf 1001}, 078 (2010)
  [arXiv:0909.2868 [hep-th]].
  %%CITATION = ARXIV:0909.2868;%%

\bibitem{Morrison:1998}
  D.~R.~Morrison and M.~R.~Plesser,
  ``Non-spherical horizons. I,''
  Adv.~Theor.~Math.~Phys. 3 (1999) 1-81,
  [arXiv:hep-th/9810201].

\bibitem{Franco:2005}
  S.~Franco, A.~Hanany, K.~D.~Kennaway, D.~Vegh and B.~Wecht,
  ``Brane dimers and quiver gauge theories,''
  JHEP 01 (2006) 096,
  [arXiv:hep-th/0504110].

\end{thebibliography}
\end{document}